\documentclass[11pt]{article}
\usepackage[utf8]{inputenc}
\usepackage[T1]{fontenc}
\usepackage{amsmath,amssymb}
\usepackage{geometry}
\usepackage{graphicx}
\usepackage{hyperref}
\usepackage{float}
\usepackage{placeins}

\usepackage{booktabs}
\usepackage{rotating}   

\usepackage{multirow}   
\usepackage{subcaption}  
\usepackage{natbib}                
\usepackage{amsmath,amssymb,amsfonts}
\usepackage{amsthm}

\newtheorem{theorem}{Theorem}[section]

\theoremstyle{plain}
\newtheorem{lemma}{Lemma}[section]
\newtheorem{proposition}{Proposition}

\newtheorem{remark}{Remark}

\title{Equilibrium and Competition in Evolutionary Dynamics}
\author{J. Medina-Diaz, F. Pe\~{n}a-Garcia and Irbin Llanqui}
\date{}

\begin{document}
\maketitle

\begin{abstract}
A fundamental problem in protobiological dynamics is to understand how chemically generated polymers can form persistent sequence distributions before the emergence of replication. We study deterministic polymer growth in which each finite sequence is followed along its genealogical structure. The system pictures an open polymerization cascade in which each polymer is produced from a unique precursor and lost by degradation and further extension. Setting fixed activated precursors, we show global well-posedness, positivity, uniqueness of a strictly positive equilibrium, and exponential convergence to an explicit steady state distribution. Under an additional uniform decay condition, this convergence becomes global exponential stability in a uniform norm.

We then couple the polymerization to a shared environmental resource with logistic growth and depletion by two activated precursors. In the resulting binary polymerization competition model, the equilibrium structure is governed by a three-dimensional core subsystem. We prove that strictly positive equilibria exist exactly above a sharp resource threshold. At the threshold the equilibrium is unique, while above it two positive branches appear. The lower branch is unstable and the upper branch is locally stable. For the complete infinite system, we exhibit positivity, global componentwise existence, a priori bounds, and under persistence and dominance assumptions, global exponential stability.

Finally, we introduce template directed replication through a replicator term. The pre-replicative equilibrium continues only under neutral fitness, and heterogeneous fitness removes it as an equilibrium of the replicated system.
\end{abstract}

\section{Introduction}
The beginning of life is a research area that has captivated the interest of both scientist and the general public even before the revolution of molecular biology \citep{Schopf24}.
Many theories has been knitted around this enigma for decades \citep{kauffman2011approaches} and even some scientist \citep{agutter2008thinking} have argued that this topic does not fall under the scope of science.
What is undeniable is that in order to understand the transition from prebiotic chemistry to the evolutionary dynamics of life, it is necessary to apply a quantitative approach to the generation, maintenance, and differential amplification of the basic elements of life: polymers.
Among the diversity of polymers, linear polymers are the most simple and will be the object of study of the present paper.

We will follow an ecological approach here in which every polymer is an individual of a dynamical population. Hence, the persistence of each polymer relies on the equilibrium between production, decay, mutation, and differential replication. This approach has bases in the work of \citet{Eigen1971} and \citet{EigenSchuster1977}. In a pre-replicative regime, activated monomers can generate diverse polymers before template-directed copying is established. This leads to a nontrivial kinetic selection after which replication can then emerge as a qualitative transition rather than as the starting point of the dynamics \citep{Nowak2008,Manapat2009}.
 
Modern origin of life chemistry increasingly frames the problem in systems terms, stressing that the relevant reaction networks must be studied under driven conditions, with sustained fluxes, environmental cycles, and kinetic bottlenecks, rather than as isolated batch reactions \citep{Sutherland2017,Pascal2013}.
This view is supported by an expanding body of experimental work on prebiotic environments driven far from equilibrium. Thermal gradients in hydrothermal pore systems can concentrate nucleotides and thereby mitigate the dilution problem faced by prebiotic synthesis \citep{Baaske2007}. More broadly, heated rock pores, drying-wetting and freeze thaw cycles, and other physical non-equilibria can promote concentration, strand separation, and sequence-dependent copying dynamics relevant to nucleic-acid chemistry \citep{Ianeselli2023}. Likewise, continuous reaction networks and compatible one-pot activation/copying scenarios have shown how prebiotic chemistry can sustain integrated pathways toward RNA precursors and nonenzymatic RNA copying in geochemically plausible settings \citep{Yi2020,Zhang2022}. At the same time, small catalytic RNA ensembles can already exhibit cooperation, competition, and frequency-dependent reproduction, indicating that network-level ecological effects may precede full Darwinian evolution \citep{Yeates2016}.

Motivated by these developments, in this work we study a deterministic, sequence-resolved model of open polymerization on the genealogy of finite $m$-ary strings. Each polymer is produced from a unique precursor, while losses arise from basal degradation and further extension to longer sequences. The resulting dynamics is an infinite system of ordinary differential equations indexed by a countable system. In contrast with standard replicator or quasispecies descriptions, the baseline model does not assume template-directed self-replication; instead, it isolates the kinetic backbone of a generative chemistry driven by a constant input of activated precursors. This allows us to ask a mathematically precise version of a protobiological question: when does an open polymerization network, by itself, settle into a robust sequence distribution, and how is that distribution altered once competition for resources and replication are added?

Our first objective is to establish the basic theory of this polymerization cascade, then couple the precursor layer to a shared environmental resource and analyze a binary competition model in which the activated roots and the competitor evolve jointly. This reduction reveals a low-dimensional ``core system'' that determines the admissible positive steady regimes of the full infinite polymer network. We characterize its equilibria, identify the threshold condition for the existence of positive steady states, and prove positivity invariance, global existence, and explicit \textit{a priori} bounds. Finally, we add a template directed replication term of replicator type and show that the previous equilibrium survives only in the neutral fitness case. 

The next section introduces the polymerization system and begins the analysis of this fundamental unbounded cascade.

\section{Dynamics of the system of polymerization.}
Let $m\ge 2$ and let $\mathcal S:=\bigcup_{n\ge 1}\{1,\dots,m\}^n$ be the set of all finite $m$-ary strings (unactivated polymers). Note that \(\mathcal S\) is countable. For \(i=(i_1,\dots,i_n)\in\mathcal S\) we write, by abuse of notation \(i=i_1\cdots i_n\) and set \(L(i)=n\) for its length and define its unique predecessor $i'$ by
\[
i' :=
\begin{cases}
i_1^* & \text{if } n=1,\\
i_1\cdots i_{n-1} & \text{if } n\ge 2,
\end{cases}
\]
where $i_1^*$ denotes the activated monomer precursor of the monomer $i_1$.
We assume the activated monomer concentrations are held constant,
\[
x_{1^*}=x_{2^*}=\cdots=x_{m^*}=1.
\]

For each \(i\in\mathcal S\), we call 
\[
i1,\dots,im
\]
the followers of \(i\). The deterministic kinetics of polymerization are given by the infinite family of linear ODEs
\begin{equation}\label{eq:polymerization}
\dot x_{i} \;=\; a_{i}\,x_{i'}-\Bigl(d+\sum_{j=1}^{m} a_{ij}\Bigr)x_{i},
\qquad i\in\mathcal S,
\end{equation}
where $a_i>0$ is the rate for the extension $i'\to i$, $a_{ij}\ge 0$ represents the loss of $i$ due to further extension to its followers, and $d>0$ is the basal degradation rate.
\begin{remark}[On the notation for monomers and finite words]
For simplicity, we label the monomer types by \(1,2,\dots,m\) and write a polymer as \(i=i_1\cdots i_n\).
If \(m\ge 10\), this notation may become ambiguous because symbols with two digits are written without separators; for example, if \(m\ge 12\), ``112'' may represent either \((1,12)\) or \((11,2)\). This is only a notational issue.
\end{remark}
To prevent any ambiguity, one may equivalently write a polymer as the finite tuple
\[
i=(i_1,\dots,i_n)
\]
instead of the shorthand \(i=i_1\cdots i_n\). All definitions and results remain unchanged. Henceforth, we keep the shorthand \(i=i_1\cdots i_n\), always interpreting it as a finite tuple of symbols, not as decimal concatenation.

Under this convention, we can state the basic dynamical properties of the polymerization system~\eqref{eq:polymerization}.

\begin{theorem}[Unique steady state and componentwise exponential convergence]
\label{thm:global-polymerization}
Let $d>0$, $a_i>0$ and $a_{ij}\ge 0$. Then:

\begin{enumerate}
\item \textbf{(Global well-posedness in the product topology and positivity)} For every initial condition
$\mathbf{x}(0)=(x_i(0))_{i\in\mathcal S}\in \mathbb R_{\ge 0}^{\mathcal S}$,
system \eqref{eq:polymerization} admits a unique global solution
$\mathbf{x}(t)\in \mathbb R_{\ge 0}^{\mathcal S}$ for all $t\ge 0$.

\item \textbf{(Unique positive equilibrium)} There exists a unique equilibrium
$\mathbf{x}^*=(x_i^*)_{i\in\mathcal S}\in \mathbb R_{>0}^{\mathcal S}$, given recursively by
\begin{equation}\label{eq:equilibrium-recursion}
x_{k}^*=\frac{a_k}{d+\sum_{j=1}^m a_{kj}}\quad (k\in\{1,\dots,m\}),
\qquad
x_i^*=\frac{a_i}{d+\sum_{j=1}^m a_{ij}}\,x_{i'}^* \quad (L(i)\ge 2).
\end{equation}
Equivalently, if $i=i_1\cdots i_n$ and $i^{(r)}:=i_1\cdots i_r$ denotes its prefix of length $r$,
\begin{equation}\label{eq:equilibrium-product}
x_i^*=\prod_{r=1}^{n}\frac{a_{i^{(r)}}}{d+\sum_{j=1}^m a_{i^{(r)}j}}.
\end{equation}

\item \textbf{(Global componentwise exponential stability)}
For every initial condition as in (1) and every $i\in\mathcal S$, there exists a constant $K_i>0$ such that
\[
|x_i(t)-x_i^*|\le K_i e^{-\delta t}\qquad (t\ge 0),
\]
where $\delta:=d/2$. In particular, $\mathbf{x}(t)\to \mathbf{x}^*$ componentwise exponentially as $t\to\infty$.
\end{enumerate}
\end{theorem}

\begin{proof}
Set
\[
k_i:=d+\sum_{j=1}^m a_{ij}\ge d,
\qquad a_i>0,
\]
so \eqref{eq:polymerization} becomes
\begin{equation}\label{eq:triangular}
\dot x_i = a_i x_{i'}-k_i x_i,\qquad i\in\mathcal S.
\end{equation}

\medskip

Fix an initial condition $\mathbf{x}(0)=(x_i(0))_{i\in\mathcal S}\in\mathbb R_{\ge 0}^{\mathcal S}$. For each $n\ge 1$, let $\mathcal S_{\le n}:=\{i\in\mathcal S:\ L(i)\le n\}$. We shall construct recursively a family of functions
\[
x^{(n)}:\ [0,\infty)\to \mathbb R^{\mathcal S_{\le n}},
\qquad
x^{(n)}(t)=\bigl(x^{(n)}_i(t)\bigr)_{i\in\mathcal S_{\le n}},
\]
such that $x^{(n)}_i(0)=x_i(0)$ for all $i\in\mathcal S_{\le n}$ and $x^{(n)}$ satisfies \eqref{eq:triangular} for every $i\in\mathcal S_{\le n}$, where the term $x_{i'}$ is interpreted as $x_{\lambda^*}\equiv 1$ whenever $i$ is a monomer. Hence, for each monomer $i$ we have the scalar linear ODE
\[
\dot x_i(t)=a_i-k_i x_i(t),
\qquad x_i(0)\ \text{given},
\]
whose unique global solution is
\begin{equation}\label{eq:monomer-solution}
x^{(1)}_i(t)=e^{-k_i t}x_i(0)+\frac{a_i}{k_i}\bigl(1-e^{-k_i t}\bigr),
\qquad t\ge 0.
\end{equation}
This defines the continuos map $x^{(1)}$ uniquely on $[0,\infty)$. Assume that for some $n\ge 1$ we have already constructed a unique continuous function $x^{(n)}$ on $[0,\infty)$.
Let $i\in\mathcal S$ with $L(i)=n+1$. Then its predecessor $i'$ has length $n$, hence
$t\mapsto x^{(n)}_{i'}(t)$ is a known continuous function on $[0,\infty)$.
Define $x^{(n+1)}_i$ as the unique solution to the nonautonomous scalar linear ODE
\begin{equation}\label{eq:length-nplus1}
\dot x_i(t)=a_i\,x^{(n)}_{i'}(t)-k_i x_i(t),
\qquad x_i(0)\ \text{given}.
\end{equation}
By the variation of constants, this solution is explicitly
\begin{equation}\label{eq:voc}
x^{(n+1)}_i(t)
=
e^{-k_i t}x_i(0)
+
a_i\int_0^t e^{-k_i(t-s)}\,x^{(n)}_{i'}(s)\,ds,
\qquad t\ge 0,
\end{equation}
and it is global because the right-hand side is well-defined for every $t\ge 0$.
For indices of length $\le n$, set $x^{(n+1)}_\ell(t):=x^{(n)}_\ell(t)$ for all $\ell\in\mathcal S_{\le n}$.
Thus we obtain a unique $x^{(n+1)}$ on $[0,\infty)$.

\medskip

If $\ell\in\mathcal S_{\le n}$, then by construction $x^{(n+1)}_\ell=x^{(n)}_\ell$ on $[0,\infty)$.
Hence, for each fixed $i\in\mathcal S$, the definition
\[
x_i(t):=x^{(L(i))}_i(t),\qquad t\ge 0,
\]
is independent of the truncation level and yields a well-defined function $x_i:[0,\infty)\to\mathbb R$. This is the \textit{consistency property}. Fix $i\in\mathcal S$ and set $n=L(i)$. Then $x_i=x^{(n)}_i$ and, if $n\ge 2$, $x_{i'}=x^{(n-1)}_{i'}$ by the consistency property (while for $n=1$ we interpret $x_{i'}\equiv 1$). Therefore $x_i$ satisfies \eqref{eq:triangular} for all $t\ge 0$.

\medskip

Let $\tilde{\mathbf{x}}(t)$ be any other family satisfying \eqref{eq:triangular} with the same initial data.
By \eqref{eq:monomer-solution}, $\tilde{x}_i=x_i$ for all monomers. Assume inductively that $\tilde{x}_{i'}=x_{i'}$ for all strings $i$ of length $\le n$. Then for any $i$ of length $n+1$, both $x_i$ and $\tilde{x}_i$ solve the same scalar linear ODE \eqref{eq:length-nplus1} with the same initial data and the same forcing term $x_{i'}(t)=\tilde{x}_{i'}(t)$. By uniqueness for scalar linear ODEs, $\tilde{x}_i=x_i$. Thus $\tilde{\mathbf{x}}=\mathbf{x}$ on $\mathcal S$. This shows existence and uniqueness of a global solution $\mathbf{x}(t)=(x_i(t))_{i\in\mathcal S}$ for all $t\ge 0$.

\medskip

Hence, nonnegativity is maintained. If $x_i(0)\ge 0$ and $x_{i'}(t)\ge 0$, then the explicit solution formula shows $x_i(t)\ge 0$ for all $t\ge 0$. Induction on length yields $\mathbf{x}(t)\in\mathbb R_{\ge 0}^{\mathcal S}$.

\medskip

Let us assume now the equilibrium values on the activated precursors by
\[
x_{\lambda^*}^*=1,\qquad \lambda=1,\dots,m.
\]
For each string $i\in\mathcal S$ of length $L(i)\ge 1$, define $x_i^*$ recursively 
\begin{equation}\label{eq:eq-rec-step2}
x_i^* := \frac{a_i}{k_i}\,x_{i'}^*.
\end{equation}

\medskip

To show that \eqref{eq:eq-rec-step2} uniquely determines a number $x_i^*$ for every $i\in\mathcal S$,
we use the fact that each sequence in $\mathcal S$ is uniquely determined by its string of prefixes. Let $i=i_1\cdots i_n\in\mathcal S$ and define its prefixes
\[
i^{(r)}:=i_1\cdots i_r,\qquad r=1,\dots,n.
\]
By construction of the predecessor map, one has $(i^{(r)})'=i^{(r-1)}$ for $r\ge 2$ and $(i^{(1)})'=i_1^*$.
Hence applying \eqref{eq:eq-rec-step2} repeatedly along the unique lineage
\[
i_1^* \to i^{(1)} \to i^{(2)} \to \cdots \to i^{(n)}=i,
\]
yields the explicit product representation
\begin{equation}\label{eq:eq-product-step2}
x_i^* = \prod_{r=1}^{n} \frac{a_{i^{(r)}}}{k_{i^{(r)}}}.
\end{equation}
In particular, the value of $x_i^*$ depends only on the unique predecessor chain of $i$, so there is no ambiguity. Since $a_{i^{(r)}}>0$ and $k_{i^{(r)}}>0$ for all $r$, \eqref{eq:eq-product-step2} implies $x_i^*>0$ for every $i\in\mathcal S$.

\medskip

Let $\mathbf{y}=(y_i)_{i\in\mathcal S}$ be any equilibrium consistent with the fixed precursor values,
i.e., assume $y_{\lambda^*}=1$ for $\lambda=1,\dots,m$ and
\[
0=a_i y_{i'}-k_i y_i,\qquad i\in\mathcal S.
\]
Then for each $i\in\mathcal S$,
\begin{equation}\label{eq:equil-relation-y}
y_i=\frac{a_i}{k_i}\,y_{i'}.
\end{equation}
We prove by induction on length $L(i)$ that $y_i=x_i^*$ for all $i\in\mathcal S$. If $L(i)=1$, then $i'=i_1^*$ for some $i_1\in\{1,\dots,m\}$, so using $y_{i_1^*}=1$ in
\eqref{eq:equil-relation-y} gives
\[
y_i=\frac{a_i}{k_i}=\frac{a_i}{k_i}\,x_{i'}^*=x_i^*.
\]
Assume now that $y_\ell=x_\ell^*$ for all strings $\ell$ of length $\le n$.
Let $i$ have length $n+1$. Then $i'$ has length $n$, hence by the induction hypothesis $y_{i'}=x_{i'}^*$.
Using \eqref{eq:equil-relation-y} and \eqref{eq:eq-rec-step2},
\[
y_i=\frac{a_i}{k_i}y_{i'}=\frac{a_i}{k_i}x_{i'}^*=x_i^*.
\]

Therefore the equilibrium is unique.

\medskip

Now, let us see the convergence part. Fix $\delta:=d/2$ and define the error $e_i(t):=x_i(t)-x_i^*$. Subtracting the steady state relation $0=a_i x_{i'}^*-k_i x_i^*$ from \eqref{eq:triangular} gives
\[
\dot e_i = a_i e_{i'}-k_i e_i.
\]
For $L(i)=1$ we have $\dot e_i=-k_i e_i$ hence $|e_i(t)|\le |e_i(0)|e^{-k_i t}\le |e_i(0)|e^{-\delta t}$. Assume that for some $n\ge 1$, all strings $i'$ with $L(i')\le n$ satisfy
$|e_{i'}(t)|\le K_{i'}e^{-\delta t}$ for all $t\ge 0$.
Let $i$ have length $n+1$. Then,
\[
e_i(t)=e^{-k_i t}e_i(0)+a_i\int_0^t e^{-k_i(t-s)}e_{i'}(s)\,ds.
\]
Using the induction hypothesis,
\[
|e_i(t)|\le |e_i(0)|e^{-k_i t}+a_i K_{i'}\int_0^t e^{-k_i(t-s)}e^{-\delta s}\,ds.
\]
Since $k_i\ge d>\delta$, the integral equals
\[
\int_0^t e^{-k_i(t-s)}e^{-\delta s}\,ds
=\frac{e^{-\delta t}-e^{-k_i t}}{k_i-\delta}
\le \frac{e^{-\delta t}}{k_i-\delta}.
\]
Therefore
\[
|e_i(t)|\le \left(|e_i(0)|+\frac{a_i K_{i'}}{k_i-\delta}\right)e^{-\delta t}.
\]
\end{proof}

Therefore, the equilibrium $\mathbf{x}^*$ given by \eqref{eq:eq-rec-step2}
is \emph{globally componentwise exponentially stable} in the following sense,
for every nonnegative initial condition $\mathbf{x}(0)\in\mathbb R_{\ge 0}^{\mathcal S}$ and every $i\in\mathcal S$,
there exists a constant $K_i>0$ such that
\[
|x_i(t)-x_i^*|\le K_i e^{-\delta t}\qquad (t\ge 0),
\]
with $\delta=d/2$.
In particular, $x_i(t)\to x_i^*$ exponentially for each $i$, and hence each coordinate $x_i(t)$ remains bounded for all $t\ge 0$.

Moreover, since $k_i=d+\sum_{j=1}^m a_{ij}\ge d$, the rate $d$ gives a uniform lower bound on the dissipation in every equation and hence yields a uniform exponential convergence rate $\delta$, independent of $i$, in the above estimate. Finally, $\mathbf{x}^*$ is the \emph{unique} equilibrium in $\mathbb R_{>0}^{\mathcal S}$ consistent with the fixed precursor values
$x_{1^*}=\cdots=x_{m^*}=1$.

\medskip

This result says that an open prebiotic polymerization network fed by a continuous input of activated monomers gives rise to a \emph{robust steady-state composition}. No matter how the initial mixture of short and long polymers is distributed, the concentration of each specific sequence \(i\) converges exponentially toward a unique steady value, which represents the long-time ``availability'' or ``background abundance'' of that sequence in a continuously driven environment (e.g., hydrothermal pores), where activation sustains the monomer supply and polymers are continually lost through degradation and further extension.
\medskip

The parameter \(d\) plays the role of a ubiquitous environmental turnover mechanism, such as hydrolysis, dilution, UV damage, or mechanical breakup, affecting every polymer regardless of sequence. Because \(k_i=d+\sum_j a_{ij}\ge d\), every equation retains at least \(d\) units of dissipation, which yields a uniform exponential convergence rate (\(\delta=d/2\)) across the whole sequence space. Thus even very long or rarely produced polymers cannot blow up in concentratio. Finally, uniqueness means the system has no alternative stable ``macrostates'' under fixed feeding so, in this linear, purely kinetic setting,
persistent changes in composition (a prerequisite for selection-like behavior) would require changing the environment/inputs (e.g., activation levels),
adding nonlinear effects (template-directed replication, catalysis), or introducing compartmentalization.

\medskip

While the previous discussion establishes \emph{componentwise} exponential convergence, in many applications it is useful to control the \emph{entire} sequence distribution at once in a uniform norm. In particular, one may naturally ask whether this convergence admits a stronger form as an exponential contraction in \(\ell^\infty(\mathcal S)\), meaning that the maximal deviation over \emph{all} polymers decays at a single rate, independent of sequence length. The next theorem provides precisely this stronger global stability statement under a natural uniform \emph{net-decay gap} condition, \(\inf_{i\in\mathcal S}(k_i-a_i)>0\), which ensures that every extension step is dominated by an intrinsic decay margin.

\begin{theorem}[$\ell^\infty$-global exponential stability under a uniform net-decay gap]
\label{thm:linfty-stability}
Consider the polymerization system \eqref{eq:triangular}. Assume the following uniform \emph{net-decay gap}:
\begin{equation}\label{eq:gap}
\gamma \;:=\; \inf_{i\in\mathcal S}(k_i-a_i)\;>\;0.
\end{equation}
Then

\begin{enumerate}
\item The equilibrium $\mathbf{x}^*=(x_i^*)_{i\in\mathcal S}$ defined by
\[
x_{\lambda^*}^*:=1\ (\lambda=1,\dots,m),\qquad x_i^*:=\frac{a_i}{k_i}\,x_{i'}^*\quad(i\in\mathcal S)
\]
satisfies $\mathbf{x}^*\in\ell^\infty(\mathcal S)$ and in fact $0<x_i^*\le 1$ for all $i\in\mathcal S$.

\item For every bounded initial condition $\mathbf{x}(0)\in\ell^\infty(\mathcal S)$
(with $\mathbf{x}(0)\ge 0$ componentwise), the corresponding solution satisfies
$\mathbf{x}(t)\in\ell^\infty(\mathcal S)$ for all $t\ge 0$ and
\begin{equation}\label{eq:linf-exp}
\|\mathbf{x}(t)-\mathbf{x}^*\|_{\infty}
\;\le\;
e^{-\gamma t}\,\|\mathbf{x}(0)-\mathbf{x}^*\|_{\infty},
\qquad t\ge 0.
\end{equation}
In particular, $\mathbf{x}^*$ is globally exponentially stable on $\ell^\infty(\mathcal S)$ with rate $\gamma$.
\end{enumerate}
\end{theorem}

\begin{proof}
From \eqref{eq:gap} we have $k_i\ge a_i+\gamma>a_i$, hence $0<\frac{a_i}{k_i}<1$ for all $i$. Iterating the recursion $x_i^*=\frac{a_i}{k_i}x_{i'}^*$ along the unique predecessor chain starting from $x_{\lambda^*}^*=1$ yields $0<x_i^*\le 1$ for all $i\in\mathcal S$, and thus
$\|\mathbf{x}^*\|_\infty\le 1$. Let $e_i(t):=x_i(t)-x_i^*$ and $\mathbf{e}(t):=\mathbf{x}(t)-\mathbf{x}^*$.
Since $x_{\lambda^*}(t)\equiv x_{\lambda^*}^*\equiv 1$, we have $e_{\lambda^*}(t)\equiv 0$ for $\lambda=1,\dots,m$. Subtracting the steady-state identity $0=a_i x_{i'}^*-k_i x_i^*$ from \eqref{eq:triangular} gives,
for all $i\in\mathcal S$,
\begin{equation}\label{eq:error}
\dot e_i \;=\; a_i e_{i'}-k_i e_i .
\end{equation}

\medskip

For \(n\ge1\), define
\[
M_n(t):=\max_{i\in\mathcal S_{\le n}} |e_i(t)|.
\]
Since \(\mathcal S_{\le n}\) is finite and each \(e_i\in C^1([0,\infty))\),
the function \(M_n\) is locally Lipschitz. Let
\[
I_n(t):=\{i\in\mathcal S_{\le n}: |e_i(t)|=M_n(t)\}.
\]

Then the standard max-Dini estimate gives
\[
D^+M_n(t)\le \max_{i\in I_n(t)} D^+|e_i(t)|.
\]

Now, for each \(i\in\mathcal S\),
\[
D^+|e_i(t)|\le a_i|e_{i'}(t)|-k_i|e_i(t)|.
\]
Indeed, if \(e_i(t)\neq0\), then
\[
\frac{d}{dt}|e_i(t)|=\operatorname{sgn}(e_i(t))\dot e_i(t)
\le a_i|e_{i'}(t)|-k_i|e_i(t)|.
\]
If \(e_i(t)=0\), then
\[
D^+|e_i(t)| = |\dot e_i(t)| = a_i|e_{i'}(t)|,
\]
which is the same inequality since \(|e_i(t)|=0\). Therefore, for any \(i\in I_n(t)\),
\[
D^+M_n(t)\le a_i|e_{i'}(t)|-k_i|e_i(t)|.
\]
Because \(i'\in\mathcal S_{\le n}\) (or \(i'=\lambda^*\), in which case \(e_{i'}=0\)),
we have \(|e_{i'}(t)|\le M_n(t)\) and \(|e_i(t)|=M_n(t)\). Hence
\[
D^+M_n(t)\le (a_i-k_i)M_n(t)\le -\gamma M_n(t).
\]
By Grönwall’s inequality for Dini derivatives,
\[
M_n(t)\le e^{-\gamma t}M_n(0),\qquad t\ge0.
\]
Since
\[
\|\mathbf e(t)\|_\infty=\sup_{i\in\mathcal S}|e_i(t)|=\sup_{n\ge1}M_n(t),
\]
and \(\mathbf x(0)\in \ell^\infty(\mathcal S)\) implies \(\mathbf e(0)\in \ell^\infty(\mathcal S)\),
we have \(M_n(0)\le \|\mathbf e(0)\|_\infty\) for every \(n\). Therefore, taking the supremum over \(n\) in
\[
M_n(t)\le e^{-\gamma t}M_n(0), \qquad t\ge0,
\]
we get
\[
\|\mathbf x(t)-\mathbf x^*\|_\infty
=
\|\mathbf e(t)\|_\infty
=
\sup_{n\ge1}M_n(t)
\le
e^{-\gamma t}\sup_{n\ge1}M_n(0)
=
e^{-\gamma t}\|\mathbf e(0)\|_\infty
=
e^{-\gamma t}\|\mathbf x(0)-\mathbf x^*\|_\infty .
\]
This is exactly \eqref{eq:linf-exp}. In particular, since \(\mathbf x^*\in \ell^\infty(\mathcal S)\), it follows that
\(\mathbf x(t)=\mathbf x^*+\mathbf e(t)\in \ell^\infty(\mathcal S)\) for all \(t\ge0\).
\end{proof}

\medskip

\medskip

Theorem~\ref{thm:linfty-stability} shows that, for the linear polymerization cascade \eqref{eq:triangular}, a uniform net decay gap enforces a strong form of robustness, the recursively defined steady profile is bounded and attracts all bounded initial data exponentially fast in $\ell^\infty(\mathcal S)$. After characterizing this fundamental behavior, we now turn to a coupled setting in which the precursor layer is no longer prescribed, but evolves dynamically together with an environmental resource variable $C$ representing the concentration of the shared competing substrate. This leads naturally to the analysis of a finite-dimensional \emph{core system}, i.e., the subsystem for the precursor concentrations and the resource variable, here $(x_{1'},x_{2'},C)$. Its equilibria determine the admissible steady regimes to which the full polymer distribution can then be attached through the same lineage recursion.
\section{Equilibrium of the ``core system"}
Now, we will analyze system \eqref{eq:5} qualitatively in order to explore conditions for the stability of the extended polymerization model.

\begin{equation}\label{eq:5}
 \begin{aligned}
        \dot x_{1'}&=1-\beta_1 Cx_{1'},\\
        \dot x_{2'}&=1-\beta_2 Cx_{2'},\\
        \dot x_{i}&= a_{i}x_{i^{'}}-(d+ a_{i1}+a_{i2})x_{i} \text{ for $i\in\mathcal{S}$} , \\
        \dot C&=C[r(1-\dfrac{C}{K})-\beta_1 x_{1'}-\beta_2 x_{2'}]. \\
\end{aligned}
\end{equation}
Here, $C=C(t)$ denotes the concentration (or abundance) of the shared environmental component regulating the dynamics, and $K>0$ is its carrying capacity in the absence of depletion by the subpopulations. The parameters $\beta_1,\beta_2>0$ are interaction coefficients. $\beta_1$ measures the strength of the coupling between $C$ and $x_{1'}$, while $\beta_2$ measures the analogous coupling between $C$ and $x_{2'}$. In particular, the terms $\beta_1 Cx_{1'}$ and $\beta_2 Cx_{2'}$ represent the $C$-mediated loss of $x_{1'}$ and $x_{2'}$, respectively, whereas $\beta_1 x_{1'}$ and $\beta_2 x_{2'}$ represent the depletion of \(C\) induced by them.

\subsection{Qualitative analysis and global stability of system \eqref{eq:5}}

Let $\mathcal S:=\bigcup_{n\ge1}\{1,2\}^n$ be the set of all finite binary strings, with the convention that each sequence $i\in\mathcal S$ has a unique precursor $i' \in \mathcal S\cup\{1',2'\}$ and two followers $i1,i2\in\mathcal S$. At steady state ($\dot x_{1'}=\dot x_{2'}=\dot x_i=\dot C=0$), the first two equations of \eqref{eq:5} give
\begin{equation}\label{eq:core-root}
x_{1'}^{*}=\frac{1}{\beta_1 C^{*}},
\qquad
x_{2'}^{*}=\frac{1}{\beta_2 C^{*}},
\end{equation}
and for every $i\in \mathcal{S}$,  the polymerization equations yield the recursive relation
\begin{equation}\label{eq:core-rec}
x_i^{*}=\frac{a_i}{d+a_{i1}+a_{i2}}\,x_{i'}^{*}.
\end{equation}

\medskip

\noindent
Equation \eqref{eq:core-rec} shows that equilibrium concentrations propagate along the unique
predecessor chain in the binary polymerization system. Notice that the state space $\mathcal S$
consists of finite words over the alphabet $\{1,2\}$ (unactivated polymers), while the activated
precursors $1',2'$ are \emph{not} elements of $\mathcal S$ but act as fixed inputs for the
polymerization recursion. More precisely, let $i=i_1\cdots i_n\in\mathcal S$ with $n\ge 1$, and define again its prefixes
\[
i^{(r)}:=i_1\cdots i_r,\qquad r=1,\dots,n.
\]
By construction of the predecessor map, we have
\[
(i^{(r)})'=i^{(r-1)}\quad\text{for } r\ge 2,
\qquad\text{and}\qquad
(i^{(1)})'=i_1'\in\{1',2'\}.
\]
Thus the unique lineage feeding $i$ is
\[
i_1' \;\to\; i^{(1)} \;\to\; i^{(2)} \;\to\; \cdots \;\to\; i^{(n)}=i.
\]
Iterating \eqref{eq:core-rec} along this chain yields the explicit product formula
\begin{equation}\label{eq:core-product}
x_i^{*}
=
\left(\prod_{r=1}^{n}\frac{a_{i^{(r)}}}{d+a_{i^{(r)}1}+a_{i^{(r)}2}}\right)\,
x_{i_1'}^{*},
\end{equation}
where the ``boundary'' equilibrium values of the activated precursors are given by \eqref{eq:core-root}. Hence, once $C^*>0$ is known, all equilibrium concentrations are uniquely determined by \eqref{eq:core-root} and the lineage recursion \eqref{eq:core-product}. In particular, positivity of parameters implies $x_i^*>0$ for all $i$ whenever $C^*>0$.

\medskip
\noindent
To determine $C^*$, set $\dot C=0$ in \eqref{eq:5}:
\[
0=C^{*}\Bigl[r\Bigl(1-\frac{C^{*}}{K}\Bigr)-\beta_1 x_{1'}^{*}-\beta_2 x_{2'}^{*}\Bigr].
\]
A steady state with finite $x_{1'}^*,x_{2'}^*$ requires $C^*>0$, and substituting \eqref{eq:core-root} yields
\[
r\Bigl(1-\frac{C^{*}}{K}\Bigr)
=
\beta_1\Bigl(\frac{1}{\beta_1 C^*}\Bigr)
+\beta_2\Bigl(\frac{1}{\beta_2 C^*}\Bigr)
=\frac{2}{C^*}.
\]
Equivalently,
\begin{equation}\label{eq:core-quadratic}
\frac{r}{K}C^{*2}-rC^*+2=0,
\end{equation}
whose discriminant is $\Delta=r^2-\frac{8r}{K}=r\bigl(r-\frac{8}{K}\bigr)$.
Therefore, \eqref{eq:core-quadratic} admits positive real solutions if and only if
\begin{equation}\label{eq:core-feasible}
rK\ge 8.
\end{equation}
If $rK=8$, there is a unique positive root,
\[
C^*=\frac{K}{2}.
\]
If $rK>8$, there are two distinct positive roots
\begin{equation}\label{eq:core-Cpm}
C_\pm^*=\frac{K}{2}\left(1\pm\sqrt{1-\frac{8}{rK}}\right),
\qquad 0<C_-^*<\frac{K}{2}<C_+^*<K.
\end{equation}

\medskip
\noindent
For each admissible choice of $C^*$ (according to whether $rK=8$ or $rK>8$), 
\[
x_{1'}^{*}= \frac{1}{\beta_{1}C^{*}},
\qquad
x_{2'}^{*}= \frac{1}{\beta_{2}C^{*}}.
\]
Moreover, for any polymer $i=i_1\cdots i_n\in\mathcal S$ ($n\ge 1$), we can write
\[
x_i^{*}
=
\left(\prod_{r=1}^{n}\frac{a_{i^{(r)}}}{d+a_{i^{(r)}1}+a_{i^{(r)}2}}\right)\,
\frac{1}{\beta_{i_1}\,C^{*}},
\]
where we use the convention $\beta_{i_1}=\beta_1$ if $i_1=1$ and $\beta_{i_1}=\beta_2$ if $i_1=2$. In short, a strictly positive steady state for the entire system exists if and only if $rK\ge 8$. It is unique when $rK=8$, while for $rK>8$ there are two distinct positive steady states
corresponding to $C^*=C_-^*$ and $C^*=C_+^*$, with all polymer concentrations determined recursively along lineages.

\medskip
\noindent
Having characterized all admissible steady states of \eqref{eq:5}, and the feasibility condition $rK\ge 8$ for their existence, the next step is to verify that the dynamics is well-defined in the biologically meaningful region of the state space. Since the variables represent concentrations, solutions must remain nonnegative whenever they start so. This motivates establishing a forward-invariance property for the positive orthant. If initially $x_{1'}(0),x_{2'}(0),C(0)>0$ and $x_i(0)>0$ for all polymers, then no component can reach zero in finite time. The following lemma formalizes this positivity invariance.

\begin{lemma}[Positivity invariance]\label{lem:positivity}
Assume $\beta_1,\beta_2,r,K, d>0$, $a_i>0$, $a_{ij}\geq 0$ and $k_i:=d+a_{i1}+a_{i2}>0$.
Let $(x_{1'},x_{2'},(x_i)_{i\in \mathcal{S}},C)$ be any (local) solution of \eqref{eq:5}
defined on $[0,T)$ with strictly positive initial data
$x_{1'}(0),x_{2'}(0),C(0)>0$ and $x_i(0)>0$ for all $i\in\mathcal{S}$. Then all components remain strictly positive on $[0,T)$. Moreover, the nonnegative orthant is forward invariant.
\end{lemma}

\begin{proof}
Write
\[
\dot C(t)=C(t)\,f(t),\qquad
f(t):=r\Bigl(1-\frac{C(t)}{K}\Bigr)-\beta_1 x_{1'}(t)-\beta_2 x_{2'}(t).
\]
Since $(x_{1'},x_{2'},C)$ is a (local) solution of the reduced subsystem, the functions $x_{1'},x_{2'},C$ are continuous on $[0,T)$,
hence so is $f$. Therefore the scalar linear ODE $\dot C=Cf$ admits the representation
\begin{equation}\label{eq:C-explicit}
C(t)=C(0)\exp\Bigl(\int_0^t f(s)\,ds\Bigr),\qquad t\in[0,T),
\end{equation}
which implies $C(t)>0$ for all $t\in[0,T)$ because $C(0)>0$ and the exponential is strictly positive.

\medskip

Since $C(t)> 0$ on $[0,T)$, each root variable solves an inhomogeneous linear ODE
\[
\dot x_{1'}(t)+\beta_1 C(t)x_{1'}(t)=1,\qquad
\dot x_{2'}(t)+\beta_2 C(t)x_{2'}(t)=1.
\]
Define the integrating factors
\[
\mu_1(t):=\exp\Bigl(\int_0^t \beta_1 C(u)\,du\Bigr),\qquad
\mu_2(t):=\exp\Bigl(\int_0^t \beta_2 C(u)\,du\Bigr),
\]
which satisfy $\mu_1(t),\mu_2(t)>0$ for all $t\in[0,T)$. Multiplying the equation for $x_{1'}$ by $\mu_1$
gives
\[
\frac{d}{dt}\bigl(\mu_1(t)x_{1'}(t)\bigr)=\mu_1(t),
\]
hence, integrating from $0$ to $t$,
\[
\mu_1(t)x_{1'}(t)=x_{1'}(0)+\int_0^t \mu_1(s)\,ds.
\]
Therefore
\begin{equation}\label{eq:x1-explicit}
x_{1'}(t)=\mu_1(t)^{-1}\left(x_{1'}(0)+\int_0^t \mu_1(s)\,ds\right)>0,
\qquad t\in[0,T),
\end{equation}
because $x_{1'}(0)>0$, $\mu_1(s)>0$, and $\mu_1(t)^{-1}>0$. The same computation yields
\begin{equation}\label{eq:x2-explicit}
x_{2'}(t)=\mu_2(t)^{-1}\left(x_{2'}(0)+\int_0^t \mu_2(s)\,ds\right)>0,
\qquad t\in[0,T).
\end{equation}

\medskip

Define the \emph{depth} $\mathrm{dep}(i)\in\mathbb N$ for $i\in\mathcal S\cup\{1',2'\}$ as the number of precursor steps
from $i$ to its (unique) ancestral root in $\{1',2'\}$. Thus $\mathrm{dep}(1')=\mathrm{dep}(2')=0$, and if $i\in\mathcal S$
then $\mathrm{dep}(i)=\mathrm{dep}(i')+1$.
We prove by induction on $n=\mathrm{dep}(i)$ that

\begin{equation}\label{eq:depth-claim}
x_i(t)>0 \quad \text{for all } t\in[0,T)\ \text{ and all } i \text{ with } \mathrm{dep}(i)=n.
\end{equation}
For $n=0$, the claim holds by \eqref{eq:x1-explicit}--\eqref{eq:x2-explicit}. Assume \eqref{eq:depth-claim} holds for all indices of depth $\le n$.
Let $i$ be any index with $\mathrm{dep}(i)=n+1$, so its precursor $i'$ has depth $n$ and therefore
$x_{i'}(t)>0$ for all $t\in[0,T)$ by the induction hypothesis. Consider the ODE for $x_i$:
\[
\dot x_i(t)+k_i x_i(t)=a_i x_{i'}(t),\qquad k_i>0,\ a_i>0.
\]
Since $x_{i'}(\cdot)$ is continuous on $[0,T)$, the function $t\mapsto a_i x_{i'}(t)$ is continuous.
Using the integrating factor $e^{k_i t}$, we obtain the formula of variation of constants
\begin{equation}\label{eq:xi-explicit}
x_i(t)=e^{-k_i t}x_i(0)+a_i\int_0^t e^{-k_i(t-s)}x_{i'}(s)\,ds,
\qquad t\in[0,T).
\end{equation}
Here $e^{-k_i t}>0$, $x_i(0)>0$, $a_i>0$, and $e^{-k_i(t-s)}x_{i'}(s)\ge 0$ for $0\le s\le t<T$ because
$x_{i'}(s)>0$ by the induction hypothesis. Hence the integral term in \eqref{eq:xi-explicit} is nonnegative, and
\[
x_i(t)\ge e^{-k_i t}x_i(0)>0\qquad \forall\,t\in[0,T).
\]
This proves \eqref{eq:depth-claim} for depth $n+1$. By induction, $x_i(t)>0$ for all $i\in\mathcal{S}$ and all
$t\in[0,T)$.
\end{proof}

\noindent
Lemma~\ref{lem:positivity} shows that the positive cone is forward invariant, so the solution cannot leave the
biologically meaningful region by hitting zero in finite time. The next issue is whether a solution that stays
positive might nevertheless \emph{blow up} in finite time by becoming unbounded. In system~\eqref{eq:5} the only
nonlinear feedback occurs through the logistic equation for the resource variable $C(t)$, which in turn modulates the decay of the activated roots $x_{1'}(t),x_{2'}(t)$ and therefore feeds the entire
polymerization tree. Hence, establishing an \emph{a priori} upper bound for $C(t)$ is the key step toward ruling out finite time blow up and guaranteeing global existence of the full componentwise solution. This is the content
of the next lemma.

\begin{lemma}[Global existence and an a priori upper bound for $C$]
\label{lem:global-existence}
Assume $\beta_1,\beta_2,r,K>0$, $a_i>0$, and $k_i>0$ for all $i\in\mathcal{S}$. For every strictly positive initial condition
\[
C(0)>0,\quad x_{1'}(0)>0,\quad x_{2'}(0)>0,\quad x_i(0)>0\ (i\in\mathcal{S}),
\]
system \eqref{eq:5} admits a unique global \emph{componentwise} solution on $[0,\infty)$.
Moreover,
\[
0<C(t)\le \max\{C(0),K\}\qquad\forall t\ge0.
\]
\end{lemma}

\begin{proof}
Consider the $3$--dimensional subsystem for $(C,x_{1'},x_{2'})$:
\[
\dot x_{1'} = 1-\beta_1 C x_{1'},\qquad
\dot x_{2'} = 1-\beta_2 C x_{2'},\qquad
\dot C = C\Bigl[r\Bigl(1-\frac{C}{K}\Bigr)-\beta_1 x_{1'}-\beta_2 x_{2'}\Bigr].
\]
Their right-hand side is locally Lipschitz on $\mathbb R^3$, hence there exists a unique maximal solution
\[
(C,x_{1'},x_{2'}) \in C^1([0,T_{\max});\mathbb R^3)
\]
for some $T_{\max}\in(0,\infty]$.  Set $M:=\max\{C(0),K\}$. We claim that $C(t)\le M$ for all $t\in[0,T_{\max})$. Since $x_{1'}(t),x_{2'}(t)\ge 0$ for all $t$ by Lemma~\ref{lem:positivity}, we have
\[
\dot C(t)
= C(t)\Bigl[r\Bigl(1-\frac{C(t)}{K}\Bigr)-\beta_1 x_{1'}(t)-\beta_2 x_{2'}(t)\Bigr]
\;\le\;
r\,C(t)\Bigl(1-\frac{C(t)}{K}\Bigr).
\]
Let $y$ be the solution of the logistic comparison equation
\[
\dot y(t)=r\,y(t)\Bigl(1-\frac{y(t)}{K}\Bigr),\qquad y(0)=C(0).
\]
Therefore, by the scalar comparison principle,
\[
0<C(t)\le y(t)\qquad \forall t\in[0,T_{\max}).
\]
On the other hand, the logistic solution $y$ with positive initial data remains positive and satisfies
\[
y(t)\le \max\{C(0),K\}\qquad \forall t\ge 0.
\]
Consequently,
\[
0<C(t)\le \max\{C(0),K\}\qquad \forall t\in[0,T_{\max}).
\]
\medskip

Since $C(t)\ge 0$, we have
\[
\dot x_{1'}(t)=1-\beta_1 C(t)x_{1'}(t)\le 1,\qquad
\dot x_{2'}(t)\le 1.
\]
Integrating on $[0,t]$ gives, for all $t\in[0,T_{\max})$,
\[
0<x_{1'}(t)\le x_{1'}(0)+t,\qquad 0<x_{2'}(t)\le x_{2'}(0)+t,
\]
so $x_{1'}$ and $x_{2'}$ remain finite on every bounded time interval.

\medskip

We have shown that $(C(t),x_{1'}(t),x_{2'}(t))$ stays bounded on $[0,T]$ for every $T<T_{\max}$.
By the standard continuation theorem for ODEs with locally Lipschitz right-hand side,
a maximal solution can fail to extend past $T_{\max}<\infty$ only if its norm blows up as $t\uparrow T_{\max}$.
Since no blow-up occurs, we must have $T_{\max}=\infty$.
Thus the core solution exists uniquely on $[0,\infty)$.

\medskip

Recall that the depth $\mathrm{dep} (i)\in\mathbb N$ is the number of predecessor steps from $i$ to its activated ancestral root in $\{1',2'\}$. Fix any $i\in \mathcal{S}$. Its ancestor chain to a root has finite length, hence once the reduced subsystem is known on $[0,\infty)$,
the precursor $x_{i'}(\cdot)$ is a known continuous function on $[0,\infty)$.
Then $x_i$ solves the linear ODE
\[
\dot x_i(t)+k_i x_i(t)=a_i x_{i'}(t),\qquad t\ge 0,
\]
and variation of constants yields the unique solution
\[
x_i(t)=e^{-k_i t}x_i(0)+a_i\int_0^t e^{-k_i(t-s)}x_{i'}(s)\,ds,\qquad t\ge0,
\]
which is $C^1$ on $[0,\infty)$. Doing this for each $i$ defines a componentwise solution
$(C,x_{1'},x_{2'},(x_i)_{i\in \mathcal{S}})$ on $[0,\infty)$.

\medskip

Let $(C,x_{1'},x_{2'},(x_i)_{i\in\mathcal{S}})$ and $(\widehat C,\widehat x_{1'},\widehat x_{2'},(\widehat x_i)_{i\in\mathcal{S}})$ be two solutions of \eqref{eq:5} on $[0,\infty)$ with the same initial data:
\[
C(0)=\widehat C(0),\quad x_{1'}(0)=\widehat x_{1'}(0),\quad x_{2'}(0)=\widehat x_{2'}(0),\quad
x_i(0)=\widehat x_i(0)\ \ \forall i\in\mathcal{S}.
\]
Based on what we said initially, we have
\[
(C(t),x_{1'}(t),x_{2'}(t))=(\widehat C(t),\widehat x_{1'}(t),\widehat x_{2'}(t))
\qquad \forall t\in[0,\infty).
\]

\medskip

We know that $\mathrm{dep}(1')=\mathrm{dep}(2')=0$, and for every $i\in\mathcal{S}$ we have
\[
\mathrm{dep}(i)=\mathrm{dep}(i')+1.
\]
We prove by induction on $n\in\mathbb N$ that
\begin{equation}\label{eq:ind-claim}
x_i(t)=\widehat x_i(t)\quad \forall t\in[0,\infty)\quad \text{for all indices } i \text{ with } \mathrm{dep}(i)\le n.
\end{equation}

\smallskip
\emph{Base case $n=0$.}
This is exactly above, i.e., the components with depth $0$ are $1',2'$ (and $C$, which already coincides),
so \eqref{eq:ind-claim} holds for $n=0$. \emph{Inductive step.} Assume \eqref{eq:ind-claim} holds for some $n\ge 0$.
Let $i$ be any index with $\mathrm{dep}(i)=n+1$, so its precursor $i'$ has depth $n$.
By the induction hypothesis,
\[
x_{i'}(t)=\widehat x_{i'}(t)\qquad \forall t\in[0,\infty).
\]
Now, for $i\in\mathcal{S}$, the component $x_i$ satisfies the linear inhomogeneous ODE
\[
\dot x_i(t)+k_i x_i(t)=a_i x_{i'}(t),\qquad t\in[0,\infty),
\]
and similarly
\[
\dot{\widehat x}_i(t)+k_i \widehat x_i(t)=a_i \widehat x_{i'}(t),\qquad t\in[0,\infty).
\]
Since $x_{i'}=\widehat x_{i'}$ on $[0,\infty)$, both $x_i$ and $\widehat x_i$ solve the \emph{same}
scalar linear ODE with the \emph{same forcing term} $t\mapsto a_i x_{i'}(t)$ and the same initial data
$x_i(0)=\widehat x_i(0)$. Scalar linear ODEs have unique solutions (e.g.\ by integrating factor),
hence
\[
x_i(t)=\widehat x_i(t)\qquad \forall t\in[0,\infty).
\]
Since $i$ was arbitrary among indices of depth $n+1$, this proves \eqref{eq:ind-claim} for $n+1$. Also, since every index $i\in\mathcal S$
has finite depth, it follows that $x_i(t)=\widehat x_i(t)$ for all $i$ and all $t\in[0,\infty)$. Hence the full family is unique.
\end{proof}


\subsection{Stability of the positive equilibria of the core subsystem}

We now analyze the local stability of the positive equilibria of the closed core subsystem
\begin{equation}\label{eq:core_subsystem_stability}
\begin{aligned}
\dot x_{1'}&=1-\beta_1 Cx_{1'},\\
\dot x_{2'}&=1-\beta_2 Cx_{2'},\\
\dot C&=
C\left[
r\left(1-\frac{C}{K}\right)
-\beta_1x_{1'}
-\beta_2x_{2'}
\right].
\end{aligned}
\end{equation}
At any positive equilibrium, the first two equations give
\[
x_{1'}^*=\frac{1}{\beta_1 C^*},
\qquad
x_{2'}^*=\frac{1}{\beta_2 C^*}.
\]
Substituting these expressions into the \(C\)-equation yields
\begin{equation}\label{eq:core_C_equilibrium}
r\left(1-\frac{C^*}{K}\right)
=
\frac{2}{C^*}.
\end{equation}
Equivalently,
\begin{equation}\label{eq:core_C_quadratic}
rC^*\left(1-\frac{C^*}{K}\right)-2=0.
\end{equation}
Thus positive equilibria exist if and only if \(rK\ge 8\). If \(rK=8\), there is a unique positive equilibrium level
\[
C^*=\frac K2.
\]
If \(rK>8\), there are two distinct positive equilibrium levels
\begin{equation}\label{eq:core_Cpm_stability}
C_\pm^*
=
\frac K2
\left(
1\pm\sqrt{1-\frac{8}{rK}}
\right),
\qquad
0<C_-^*<\frac K2<C_+^*<K.
\end{equation}

\medskip

We now classify the local stability of these equilibria for the full core subsystem \eqref{eq:core_subsystem_stability}. Let
\[
E_c^*=(x_{1'}^*,x_{2'}^*,C^*)
\]
be a positive core equilibrium. The Jacobian matrix of \eqref{eq:core_subsystem_stability}, written in the order
\((x_{1'},x_{2'},C)\), is
\[
J(C^*)
=
\begin{pmatrix}
-\beta_1 C^* & 0 & -\beta_1x_{1'}^*\\
0 & -\beta_2 C^* & -\beta_2x_{2'}^*\\
-\beta_1 C^* & -\beta_2 C^* &
r\left(1-\frac{C^*}{K}\right)
-\beta_1x_{1'}^*
-\beta_2x_{2'}^*
-\frac{r}{K}C^*
\end{pmatrix}.
\]
Using the equilibrium identities
\[
\beta_1x_{1'}^*=\frac1{C^*},
\qquad
\beta_2x_{2'}^*=\frac1{C^*},
\]
and
\[
r\left(1-\frac{C^*}{K}\right)
-\beta_1x_{1'}^*
-\beta_2x_{2'}^*
=0,
\]
this becomes
\begin{equation}\label{eq:core_jacobian}
J(C^*)
=
\begin{pmatrix}
-\beta_1 C^* & 0 & -\frac1{C^*}\\
0 & -\beta_2 C^* & -\frac1{C^*}\\
-\beta_1 C^* & -\beta_2 C^* & -\frac{r}{K}C^*
\end{pmatrix}.
\end{equation}

Set
\[
p:=\beta_1 C^*,\qquad
q:=\beta_2 C^*,\qquad
\alpha:=\frac{r}{K}C^*.
\]
Then \(p,q,\alpha>0\), and the characteristic polynomial of \(J(C^*)\) is
\[
P(\lambda)
=
\det(\lambda I-J(C^*))
=
\lambda^3+A_1\lambda^2+A_2\lambda+A_3,
\]
where
\[
A_1=p+q+\alpha,
\]
\[
A_2=pq+(p+q)\left(\alpha-\frac1{C^*}\right),
\]
and
\[
A_3=pq\left(\alpha-\frac{2}{C^*}\right).
\]
Using the equilibrium identity \eqref{eq:core_C_equilibrium}, we get
\[
\frac{2}{C^*}=r\left(1-\frac{C^*}{K}\right),
\]
and therefore
\[
\alpha-\frac{2}{C^*}
=
\frac{r}{K}C^*
-
r\left(1-\frac{C^*}{K}\right)
=
\frac{r}{K}(2C^*-K).
\]
Hence
\begin{equation}\label{eq:A3_sign}
A_3
=
pq\,\frac{r}{K}(2C^*-K).
\end{equation}

If \(C^*=C_-^*\), then \(C_-^*<K/2\), and therefore \(A_3<0\). Since \(P(0)=A_3<0\) while \(P(\lambda)\to+\infty\) as \(\lambda\to+\infty\), the characteristic polynomial has at least one positive real root. Thus the lower equilibrium branch \(C_-^*\) is unstable for the core subsystem.

Now suppose \(C^*=C_+^*\). Then \(C_+^*>K/2\), so \(A_3>0\). Moreover,
\[
\alpha-\frac1{C^*}
=
\frac{r}{K}C^*-\frac1{C^*}.
\]
Using again \eqref{eq:core_C_equilibrium}, we may write
\[
\frac1{C^*}
=
\frac r2\left(1-\frac{C^*}{K}\right),
\]
and therefore
\[
\alpha-\frac1{C^*}
=
\frac{r}{K}C^*
-
\frac r2\left(1-\frac{C^*}{K}\right)
=
\frac{r}{2K}(3C^*-K).
\]
Since \(C_+^*>K/2\), we have \(3C_+^*-K>0\). Hence
\[
A_2>0.
\]
Clearly \(A_1>0\). It remains to verify the Routh--Hurwitz condition
\[
A_1A_2>A_3.
\]
Indeed,
\[
A_1A_2-A_3
=
(p+q+\alpha)
\left[
pq+(p+q)\left(\alpha-\frac1{C^*}\right)
\right]
-
pq\left(\alpha-\frac2{C^*}\right).
\]
Expanding and simplifying gives
\[
A_1A_2-A_3
=
(p+q)pq
+
\frac{2pq}{C^*}
+
(p+q)(p+q+\alpha)
\left(\alpha-\frac1{C^*}\right).
\]
For \(C^*=C_+^*>K/2\), all three terms on the right-hand side are strictly positive. Therefore
\[
A_1A_2>A_3.
\]
By the Routh--Hurwitz criterion for cubic polynomials, all eigenvalues of \(J(C_+^*)\) have strictly negative real parts. Hence the upper equilibrium branch \(C_+^*\) is locally asymptotically stable for the full core subsystem.

In the threshold case \(rK=8\), the only positive equilibrium has \(C^*=K/2\). In this case \(A_3=0\), so the Jacobian has a zero eigenvalue. Hence the equilibrium is non-hyperbolic, and linearization alone does not determine its stability.

\medskip

Consequently, for the closed three-dimensional core subsystem, the lower positive branch is unstable and the upper positive branch is locally asymptotically stable whenever \(rK>8\). Since the core subsystem is closed inside the full polymerization--competition model, any Lyapunov-stable equilibrium of the full infinite system must project onto a Lyapunov-stable equilibrium of the core. Therefore, the full equilibrium associated with \(C_-^*\) cannot be Lyapunov stable, while the branch associated with \(C_+^*\) is the only positive branch that can serve as a locally stable candidate for the full system.

\begin{lemma}[Core stability is necessary for full-system stability]
\label{lem:reduced_vs_full}
Let 
\[
E^*=\bigl(C^*,x_{1'}^*,x_{2'}^*,(x_i^*)_{i\in\mathcal S}\bigr)
\]
be a strictly positive equilibrium of the full system \eqref{eq:5}. Consider the $\ell^2$-type norm (whenever finite)
\[
\|(C,x)-(C^*,x^*)\|^2
:= |C-C^*|^2 + |x_{1'}-x_{1'}^*|^2 + |x_{2'}-x_{2'}^*|^2
   + \sum_{i\in\mathcal S} |x_i-x_i^*|^2,
\]
where $x=\bigl(C,x_{1'},x_{2'},(x_i)_{i\in\mathcal S}\bigr)$. If $E^*$ is Lyapunov stable for the full system \eqref{eq:5} with respect to this topology, then the core equilibrium $(C^*,x_{1'}^*,x_{2'}^*)$ is Lyapunov stable for the
closed core subsystem
\begin{equation}\label{eq:core_subsystem}
\dot x_{1'}=1-\beta_1 Cx_{1'},\qquad
\dot x_{2'}=1-\beta_2 Cx_{2'},\qquad
\dot C=C\Bigl[r\Bigl(1-\frac{C}{K}\Bigr)-\beta_1 x_{1'}-\beta_2 x_{2'}\Bigr].
\end{equation}
In particular, if $(C^*,x_{1'}^*,x_{2'}^*)$ is unstable for \eqref{eq:core_subsystem}, then $E^*$ cannot be
Lyapunov stable for the full system \eqref{eq:5}.
\end{lemma}

\begin{proof}
Inspecting \eqref{eq:5}, the right-hand sides of the equations for $(C,x_{1'},x_{2'})$ depend only on
$(C,x_{1'},x_{2'})$ and not on the remaining variables $(x_i)_{i\in\mathcal S}$. Hence any solution of the full
system projects onto a solution of the closed core subsystem \eqref{eq:core_subsystem}. Assume $E^*$ is Lyapunov stable for the full system with respect to $\|\cdot\|$. Fix $\varepsilon>0$. Then there exists $\delta>0$ such that for every initial condition
\[
(C(0),x_{1'}(0),x_{2'}(0),(x_i(0))_{i\in\mathcal S})
\quad\text{with}\quad
\|(C(0),x(0))-(C^*,x^*)\|<\delta,
\]
the corresponding solution satisfies
\[
\|(C(t),x(t))-(C^*,x^*)\|<\varepsilon
\qquad\text{for all }t\ge 0.
\]
Moreover, since $E^*$ is strictly positive, we restrict to strictly positive initial data. Then Lemma~\ref{lem:global-existence} guarantees that the resulting solution of \eqref{eq:5} exists globally on
$[0,\infty)$, so the above estimate is meaningful for every $t\ge 0$\\

Now restrict to perturbations affecting only the core variables: choose initial data such that
\[
x_i(0)=x_i^*\quad\text{for all }i\in\mathcal S,
\]
and $(C(0),x_{1'}(0),x_{2'}(0))$ is $\delta$--close to $(C^*,x_{1'}^*,x_{2'}^*)$ in $\mathbb R^3$.
Then
\[
\|(C(0),x(0))-(C^*,x^*)\|^2
=|C(0)-C^*|^2+|x_{1'}(0)-x_{1'}^*|^2+|x_{2'}(0)-x_{2'}^*|^2.
\]
By the full-system stability hypothesis, for all $t\ge 0$ we have
\[
|C(t)-C^*|^2+|x_{1'}(t)-x_{1'}^*|^2+|x_{2'}(t)-x_{2'}^*|^2
\le \|(C(t),x(t))-(C^*,x^*)\|^2<\varepsilon^2.
\]
Therefore the projected core trajectory $(C(t),x_{1'}(t),x_{2'}(t))$ remains $\varepsilon$--close to
$(C^*,x_{1'}^*,x_{2'}^*)$ for all $t\ge 0$, and since it solves \eqref{eq:core_subsystem}, the core equilibrium is
Lyapunov stable for \eqref{eq:core_subsystem}.
\end{proof}

\begin{remark}[Extension to the full system requires additional dissipativity]\label{rem:extension_needs_lyapunov}
The reduced classification does \emph{not} by itself prove global asymptotic stability of the full infinite system \eqref{eq:5}. To obtain global asymptotic stability (and exponential decay in an $\ell^2$-framework),
one standard approach combines:
\begin{itemize}
\item invariance/positivity and global existence of solutions,
\item a uniform lower bound $C(t)\ge \underline C>0$ and an a priori upper bound $C(t)\le \overline C$,
\item a coercive Lyapunov functional $V$ controlling the $\ell^2$--distance to equilibrium,
\item dominance/Young inequalities ensuring $\dot V\le -\beta V$ for some $\beta>0$.
\end{itemize}

\end{remark}

\noindent

Motivated by Remark~\ref{rem:extension_needs_lyapunov}, we now verify one of the structural ingredients needed for a global stability result. Beyond positivity and global existence, a key missing ingredient is a \emph{uniform persistence} estimate for the competitor variable \(C\). Namely, we need to prevent \(C(t)\) from approaching \(0\), since such a degeneration would typically destroy coercivity and uniform dissipation in the \(\ell^2\)-estimates.

The next lemma provides precisely this type of forward-invariance statement. Under an explicit balance condition between the logistic growth term and suitable upper bounds for the root concentrations, any initial lower bound of the form \(C(0)\ge \underline C\) is propagated for as long as the core trajectory exists.
\begin{lemma}[Forward invariance of a positive lower bound for the competitor]
\label{lem:C-lower-invariant}
Let $(C,x_{1'},x_{2'})$ be a (local) $C^{1}$ solution of the core subsystem \eqref{eq:core_subsystem}
on $[0,T)$, with parameters $\beta_1,\beta_2,r,K>0$ and strictly positive initial data.
Fix $\underline C\in(0,K)$ and assume

\begin{equation}\label{eq:H0}
C(0)\ge \underline C,
\qquad
r\Bigl(1-\frac{\underline C}{K}\Bigr)\ \ge\ \beta_1 M_1+\beta_2 M_2,
\end{equation}
where
\[
M_k:=\max\Bigl\{x_{k'}(0),\frac{1}{\beta_k\underline C}\Bigr\},\qquad k=1,2.
\]
Then
\[
C(t)\ge \underline C \qquad \forall\,t\in[0,T).
\]
\end{lemma}

\begin{proof}
Consider the closed set
\[
\Omega:=\bigl\{(x,y,z)\in\mathbb R^3:\ x\ge \underline C,\ 0\le y\le M_1,\ 0\le z\le M_2\bigr\}.
\]
By construction we have $(C(0),x_{1'}(0),x_{2'}(0))\in\Omega$.
Assume for contradiction that the trajectory leaves $\Omega$.
Define the "first exit time"
\[
A:=\{t\in(0,T):\, y(t)\notin\Omega\}\neq\varnothing,
\qquad
t_*:=\inf A.
\]

\medskip
Since $(C,x_{1'},x_{2'})$ is a $C^{1}$ solution on $[0,T)$, the trajectory
\[
y(t):=(C(t),x_{1'}(t),x_{2'}(t))
\]
is continuous on $[0,T)$. Observe that $y(0)\in\Omega$. Then $0<t_*<T$. Moreover, for every $t\in[0,t_*)$ we must have $y(t)\in\Omega$; otherwise, if
$y(t_0)\notin\Omega$ for some $t_0<t_*$, then $t_0\in A$ and hence
$t_*=\inf A\le t_0<t_*$, a contradiction. Thus $y(t)\in\Omega$ for all $t<t_*$.\\

To locate $y(t_*)$, let $t_n\uparrow t_*$ with $t_n<t_*$. Then $y(t_n)\in\Omega$ for all $n$,
and by continuity $y(t_n)\to y(t_*)$. Since $\Omega$ is closed, it follows that $y(t_*)\in\Omega$. We observe that $y(t_*)\notin \operatorname{int}(\Omega)$. Indeed, if $y(t_*)\in \operatorname{int}(\Omega)$,
there exists $\varepsilon>0$ such that $B_\varepsilon(y(t_*))\subset\Omega$. By continuity,
there exists $\delta>0$ such that $|t-t_*|<\delta$ implies $y(t)\in B_\varepsilon(y(t_*))\subset\Omega$,
so in particular $y(t)\in\Omega$ for all $t\in(t_*,t_*+\delta)$, contradicting the definition of
$t_*=\inf A$. Therefore $y(t_*)\in \partial\Omega$.\\

Since $y(t_*)=(C(t_*),x_{1'}(t_*),x_{2'}(t_*))\in\partial\Omega$, at least one of the defining
inequalities of $\Omega$ is saturated at $t_*$. Hence,
\[
C(t_*)=\underline C,\quad \text{or}\quad x_{1'}(t_*)\in\{0,M_1\},\quad \text{or}\quad x_{2'}(t_*)\in\{0,M_2\}.
\]
By Lemma~\ref{lem:positivity}, $x_{1'}(t),x_{2'}(t)>0$ for all $t\in[0,T)$. Therefore, at the first exit time we must have
\[
C(t_*)=\underline C,\qquad \text{or}\qquad x_{1'}(t_*)=M_1,\qquad \text{or}\qquad x_{2'}(t_*)=M_2.
\]

\smallskip
\noindent\emph{Case 1: $x_{k'}(t_*)=M_k$ for some $k\in\{1,2\}$.}
Since $C(t)\ge \underline C$ on $[0,t_*]$, at $t=t_*$ we have
\[
\dot x_{k'}(t_*)=1-\beta_k C(t_*)x_{k'}(t_*)
\le 1-\beta_k \underline C\,M_k \le 0,
\]
because $M_k\ge 1/(\beta_k\underline C)$ by definition.
Thus $x_{k'}$ cannot cross \emph{upwards} through the level $M_k$ at $t_*$.

\smallskip
\noindent\emph{Case 2: $C(t_*)=\underline C$.}
Since the trajectory stays in $\Omega$ up to $t_*$, we also have
$x_{1'}(t_*)\le M_1$ and $x_{2'}(t_*)\le M_2$. Using the $C$--equation,
\[
\dot C(t_*)
= C(t_*)\Bigl[r\Bigl(1-\frac{C(t_*)}{K}\Bigr)-\beta_1 x_{1'}(t_*)-\beta_2 x_{2'}(t_*)\Bigr]
\ge \underline C\Bigl[r\Bigl(1-\frac{\underline C}{K}\Bigr)-\beta_1 M_1-\beta_2 M_2\Bigr]\ge 0
\]
by \eqref{eq:H0}. Hence $C$ cannot cross \emph{downwards} through the level $\underline C$ at $t_*$.\\

Once $y(t)$ reaches the boundary $\partial\Omega$ at $t_*$, in any case above, $y(t)$ cannot cross it, contradicting that $t_*$ is the first exit time from $\Omega$. Therefore the trajectory never leaves $\Omega$, and in particular
$C(t)\ge \underline C$ for all $t\in[0,T)$.
\end{proof}

Lemma~\ref{lem:C-lower-invariant} provides the key \emph{forward-invariant positivity buffer} needed in the global Lyapunov analysis. Under \eqref{eq:H0} the competitor level cannot fall below a fixed threshold $\underline C$ along the core dynamics. Together with the a priori upper bound $C(t)\le \overline C=\max\{C(0),K\}$ from the logistic comparison, this yields 
\[
\underline C\le C(t)\le \overline C \qquad \forall t\ge 0,
\]
which prevent degeneracy in the dissipation estimates and allow all coupling terms involving $C(t)$ to be controlled by constants depending only on $\underline C$ and $\overline C$. With these bounds in hand, we are able to impose the dominance conditions (35)--(36) and construct an $\ell^2$--Lyapunov functional whose derivative is strictly negative, leading to global asymptotic stability with exponential convergence for the full infinite system, as stated next. Before the next Theorem, we define

\[
\|\tilde x(t)\|_{\ell^2}^2+|\tilde C(t)|^2
:=
\bigl(x_{1'}(t)-x_{1'}^*\bigr)^2+\bigl(x_{2'}(t)-x_{2'}^*\bigr)^2
+\sum_{i\in\mathcal{S}}\bigl(x_i(t)-x_i^*\bigr)^2
+\bigl(C(t)-C^*\bigr)^2
<\infty.
\]

\begin{theorem}[Exponential stability in \(\ell^2\) under persistence and dominance conditions]
\label{thm:core-GAS}
Consider the infinite system \eqref{eq:5} with two kinds of active monomers, and assume
\[
\beta_1,\beta_2,r,K>0,\qquad a_i>0,\qquad 
k_i:=d+a_{i1}+a_{i2}>0\quad (i\in\mathcal{S}).
\]
Assume $rK> 8$ and fix the equilibrium competitor level $C^*>0$ as
\[
C^*:=\frac{K}{2}\Bigl(1+\sqrt{1-\frac{8}{rK}}\Bigr),
\qquad \frac K2<C^*<K.
\]

Define the associated positive equilibrium coordinates by
\[
x_{1'}^*=\frac{1}{\beta_1 C^*},\qquad
x_{2'}^*=\frac{1}{\beta_2 C^*},\qquad
x_i^*=\frac{a_i}{k_i}\,x_{i'}^*\quad (i\in\mathcal{S}).
\]

Define

\begin{equation}\label{eq:H1}
a_{\max}:=\sup_{i\in\mathcal{S}} a_i<\infty,
\qquad
k_{\min}:=\inf_{i\in\mathcal{S}} k_i > \frac{3}{2}a_{\max}.
\end{equation}

Let $(x_{1'},x_{2'},(x_i)_{i\in\mathcal{S}},C)$ be the unique global solution on $[0,\infty)$
issued from strictly positive initial data (cf.\ Lemma~\ref{lem:global-existence}).
Assume moreover that $H0$ of Lemma~\ref{lem:C-lower-invariant} holds for this trajectory. Assume the \emph{equilibrium dominance} inequalities

\begin{equation}\label{eq:H2a}
2\beta_1 \underline C>a_{\max},\qquad 2\beta_2 \underline C>a_{\max},    
\end{equation}

and
\begin{equation}\label{eq:H2b}
\frac{\bigl[\beta_1(\overline C+x_{1'}^*)\bigr]^2}{2\beta_1 \underline C-a_{\max}}
+
\frac{\bigl[\beta_2(\overline C+x_{2'}^*)\bigr]^2}{2\beta_2 \underline C-a_{\max}}
\;<\;
\frac{2r}{K}\,\underline C,    
\end{equation}
where
\[
\overline C:=\max\{C(0),K\}.
\]

Assume the initial deviation is square-summable:
\[
\|\tilde x(0)\|_{\ell^2}^2+|\tilde C(0)|^2
=
\bigl(x_{1'}(0)-x_{1'}^*\bigr)^2+\bigl(x_{2'}(0)-x_{2'}^*\bigr)^2
+\sum_{i\in\mathcal{S}}\bigl(x_i(0)-x_i^*\bigr)^2
+\bigl(C(0)-C^*\bigr)^2
<\infty.
\]
Then the equilibrium $(C^*,x_{1'}^*,x_{2'}^*,(x_i^*)_{i\in\mathcal{S}})$ is globally asymptotically stable in the sense that
\[
\|\tilde x(t)\|_{\ell^2}^2+|\tilde C(t)|^2 \to 0
\qquad (t\to\infty),
\]
and in fact there exists $\beta>0$ (depending only on the constants in (33)--(36)) such that
\[
\|\tilde x(t)\|_{\ell^2}^2+|\tilde C(t)|^2
\;\le\;
\bigl(\|\tilde x(0)\|_{\ell^2}^2+|\tilde C(0)|^2\bigr)e^{-\beta t},
\qquad t\ge 0.
\]
\end{theorem}

\begin{proof}
Set
\[
\tilde C:=C-C^*,\qquad \tilde x_{1'}:=x_{1'}-x_{1'}^*,\qquad \tilde x_{2'}:=x_{2'}-x_{2'}^*,
\qquad \tilde x_i:=x_i-x_i^* \ (i\in\mathcal{S}).
\]
Using the equilibrium relations
\[
1=\beta_1 C^*x_{1'}^*,\qquad 1=\beta_2 C^*x_{2'}^*,
\qquad
0=a_i x_{i'}^*-k_i x_i^* \ \ (i\in\mathcal{S}),
\]
we obtain for $i\in\mathcal{S}$:

\begin{equation}\label{equation:G1}
\dot{\tilde x}_i = a_i \tilde x_{i'}-k_i \tilde x_i.
\end{equation}

For the roots:
\begin{equation}\label{equation:G2a}
\dot{\tilde x}_{1'}= \bigl(1-\beta_1 Cx_{1'}\bigr)-\bigl(1-\beta_1 C^*x_{1'}^*\bigr)
= -\beta_1\bigl(Cx_{1'}-C^*x_{1'}^*\bigr)
= -\beta_1\bigl(C\tilde x_{1'}+x_{1'}^*\tilde C\bigr),
\end{equation}

and similarly

\begin{equation}\label{equation:G2b}
\dot{\tilde x}_{2'}= -\beta_2\bigl(C\tilde x_{2'}+x_{2'}^*\tilde C\bigr).
\end{equation}

For $C$:
\[
\dot{\tilde C}=\dot C
= C\Bigl[r\Bigl(1-\frac{C}{K}\Bigr)-\beta_1 x_{1'}-\beta_2 x_{2'}\Bigr]
- C\underbrace{\Bigl[r\Bigl(1-\frac{C^*}{K}\Bigr)-\beta_1 x_{1'}^*-\beta_2 x_{2'}^*\Bigr]}_{=0},
\]
hence
\begin{equation}\label{equation:G3}
\dot{\tilde C}= C\Bigl[-\frac{r}{K}\tilde C-\beta_1\tilde x_{1'}-\beta_2\tilde x_{2'}\Bigr].
\end{equation}

\medskip

Let $\mathcal T_N$ denote the finite set of nodes up to depth $N$ (including $1',2'$).
Define
\[
V_N(t):=\tilde C(t)^2+\tilde x_{1'}(t)^2+\tilde x_{2'}(t)^2
+\sum_{\substack{i\in\mathcal T_N\\ i\in\mathcal{S}}}\tilde x_i(t)^2.
\]
Since $V_N$ is a finite sum of $C^1$ functions, it is $C^1$ and
\[
\dot V_N = 2\tilde C\dot{\tilde C}+2\tilde x_{1'}\dot{\tilde x}_{1'}+2\tilde x_{2'}\dot{\tilde x}_{2'}
+\sum_{\substack{i\in\mathcal T_N\\ i\in\mathcal{S}}}2\tilde x_i\dot{\tilde x}_i.
\]

\medskip

Using (40),

\begin{equation}\label{equation: G4}
2\tilde C\dot{\tilde C}
=2\tilde C\,C\Bigl[-\frac{r}{K}\tilde C-\beta_1\tilde x_{1'}-\beta_2\tilde x_{2'}\Bigr]
= -2\frac{r}{K}C\,\tilde C^2
-2\beta_1 C\,\tilde C\,\tilde x_{1'}
-2\beta_2 C\,\tilde C\,\tilde x_{2'}.
\end{equation}

Using (38)--(39),
\begin{equation}\label{equation: G5a}
2\tilde x_{1'}\dot{\tilde x}_{1'}
=2\tilde x_{1'}\bigl[-\beta_1(C\tilde x_{1'}+x_{1'}^*\tilde C)\bigr]
= -2\beta_1 C\,\tilde x_{1'}^2 -2\beta_1 x_{1'}^*\tilde C\,\tilde x_{1'},
\end{equation}

\begin{equation}\label{equation: G5b}
2\tilde x_{2'}\dot{\tilde x}_{2'}
= -2\beta_2 C\,\tilde x_{2'}^2 -2\beta_2 x_{2'}^*\tilde C\,\tilde x_{2'}.
\end{equation}
Finally, using (37),
\begin{equation}\label{equation: G6}
2\tilde x_i\dot{\tilde x}_i = 2\tilde x_i(a_i\tilde x_{i'}-k_i\tilde x_i)
= -2k_i\tilde x_i^2 + 2a_i\tilde x_i\tilde x_{i'}.
\end{equation}

Summing (41)--(44) gives
\begin{align}\label{equation: G7}
\dot V_N
&=
-2\frac{r}{K}C\,\tilde C^2
-2\beta_1 C\,\tilde x_{1'}^2
-2\beta_2 C\,\tilde x_{2'}^2
-\sum_{\substack{i\in\mathcal T_N\\ i\in\mathcal{S}}}2k_i\tilde x_i^2
+\sum_{\substack{i\in\mathcal T_N\\ i\in\mathcal{S}}}2a_i\tilde x_i\tilde x_{i'}
\nonumber\\
&\quad
-2\beta_1(C+x_{1'}^*)\tilde C\,\tilde x_{1'}
-2\beta_2(C+x_{2'}^*)\tilde C\,\tilde x_{2'}.
\end{align}

\medskip

For each $i\in\mathcal{S}$, $2a_i\tilde x_i\tilde x_{i'} \le a_i\tilde x_i^2+a_i\tilde x_{i'}^2 \le a_{\max}\tilde x_i^2 + a_i\tilde x_{i'}^2.$ Summing over $i\in\mathcal T_N\setminus\{1',2'\}$,
\[
\sum_{\substack{i\in\mathcal T_N\\ i\in\mathcal{S}}}2a_i\tilde x_i\tilde x_{i'}
\le a_{\max}\sum_{\substack{i\in\mathcal T_N\\ i\in\mathcal{S}}}\tilde x_i^2 + \sum_{\substack{i\in\mathcal T_N\\ i\in\mathcal{S}}}a_i\tilde x_{i'}^2.
\]
Because the tree is binary, each node $j$ has at most two children, hence
\[
\sum_{\substack{i\in\mathcal T_N\\ i\in\mathcal{S}}} a_i\tilde x_{i'}^2
\le 2a_{\max}\sum_{\substack{i\in\mathcal T_N\\ i\in\mathcal{S}}}\tilde x_j^2 + a_{\max}\bigl(\tilde x_{1'}^2+\tilde x_{2'}^2\bigr).
\]
Therefore,
\begin{equation}\label{equation: G8}
\sum_{\substack{i\in\mathcal T_N\\ i\in\mathcal{S}}} 2a_i\tilde x_i\tilde x_{i'}
\le 3a_{\max}\sum_{\substack{j\in\mathcal T_N\\ j\in\mathcal{S}}}\tilde x_j^2 + a_{\max}\bigl(\tilde x_{1'}^2+\tilde x_{2'}^2\bigr).
\end{equation}
Combining (46) with $k_i\ge k_{\min}$ yields
\begin{equation}\label{equation: G9}
-\sum_{\substack{i\in\mathcal T_N\\ i\in\mathcal{S}}}2k_i\tilde x_i^2 + \sum_{\substack{i\in\mathcal T_N\\ i\in\mathcal{S}}}2a_i\tilde x_i\tilde x_{i'}
\le
-\bigl(2k_{\min}-3a_{\max}\bigr)\sum_{\substack{i\in\mathcal T_N\\ i\in\mathcal{S}}}\tilde x_i^2
+ a_{\max}\bigl(\tilde x_{1'}^2+\tilde x_{2'}^2\bigr).
\end{equation}

Set
\[
\eta:=2k_{\min}-3a_{\max}.
\]
Then by (34), $\eta>0$. Insert (47) into (45). Using (33) and by Lemma 1.4, $C(t)\ge \underline C$ and $C(t)\le \overline C$ for all $t$. Thus for $k=1,2$,
\begin{equation}\label{equation: G10}
-2\beta_k C\,\tilde x_{k'}^2 + a_{\max}\tilde x_{k'}^2
\le -(2\beta_k \underline C-a_{\max})\,\tilde x_{k'}^2.
\end{equation}

\medskip

Fix a parameter $\theta\in(0,1)$.
For $k=1,2$,
\[
2\beta_k|C+x_{k'}^*|\,|\tilde C|\,|\tilde x_{k'}|
\le 2\beta_k(\overline C+x_{k'}^*)\,|\tilde C|\,|\tilde x_{k'}|.
\]
Apply Young's inequality in the form
\[
2ab\le \theta\,\xi\,b^2+\frac{a^2}{\theta\,\xi}
\qquad(\theta\in(0,1),\ \xi>0),
\]
with
\[
a:=\beta_k(\overline C+x_{k'}^*)|\tilde C|,
\qquad
b:=|\tilde x_{k'}|,
\qquad
\xi_k:=2\beta_k\underline C-a_{\max}>0
\quad\text{(by (35)).}
\]
This gives
\begin{equation}\label{equation: 11}
2\beta_k(\overline C+x_{k'}^*)\,|\tilde C|\,|\tilde x_{k'}|
\le
\theta(2\beta_k\underline C-a_{\max})\,\tilde x_{k'}^2
+
\frac{\bigl[\beta_k(\overline C+x_{k'}^*)\bigr]^2}{\theta(2\beta_k\underline C-a_{\max})}\tilde C^2.
\end{equation}

Combining (48) with (49) yields, for each $k=1,2$,
\begin{equation}\label{equation: G12}
-(2\beta_k\underline C-a_{\max})\,\tilde x_{k'}^2
+2\beta_k(\overline C+x_{k'}^*)\,|\tilde C|\,|\tilde x_{k'}|
\le
-(1-\theta)(2\beta_k\underline C-a_{\max})\,\tilde x_{k'}^2
+
\frac{\bigl[\beta_k(\overline C+x_{k'}^*)\bigr]^2}{\theta(2\beta_k\underline C-a_{\max})}\tilde C^2.
\end{equation}

\medskip

Using (45), (48), and (50), and also $C\ge \underline C$ in the term
$-2\frac{r}{K}C\,\tilde C^2$, we get
\begin{align*}
\dot V_N(t)
&\le
-2\frac{r}{K}\underline C\,\tilde C(t)^2
-\eta\sum_{\substack{i\in\mathcal T_N\\ i\in\mathcal{S}}}\tilde x_i(t)^2
-\sum_{k=1}^2 (1-\theta)(2\beta_k\underline C-a_{\max})\,\tilde x_{k'}(t)^2
\nonumber\\
&\quad
+\sum_{k=1}^2
\frac{\bigl[\beta_k(\overline C+x_{k'}^*)\bigr]^2}{\theta(2\beta_k\underline C-a_{\max})}\,\tilde C(t)^2.
\end{align*}
Grouping the $\tilde C^2$ terms gives

\begin{equation}\label{equation: G14}
\dot V_N(t)
\leq
-\gamma_C(\theta)\,\tilde C(t)^2
-\eta\sum_{\substack{i\in\mathcal T_N\\ i\in\mathcal{S}}}\tilde x_i(t)^2
-\lambda_1(\theta)\,\tilde x_{1'}(t)^2
-\lambda_2(\theta)\,\tilde x_{2'}(t)^2,
\end{equation}

where
\[
\lambda_k(\theta):=(1-\theta)(2\beta_k\underline C-a_{\max})>0\qquad (k=1,2),
\]
and
\[
\gamma_C(\theta):=
\frac{2r}{K}\underline C
-
\frac{\bigl[\beta_1(\overline C+x_{1'}^*)\bigr]^2}{\theta(2\beta_1\underline C-a_{\max})}
-
\frac{\bigl[\beta_2(\overline C+x_{2'}^*)\bigr]^2}{\theta(2\beta_2\underline C-a_{\max})}.
\]
Set
\[
S:=\frac{\bigl[\beta_1(\overline C+x_{1'}^*)\bigr]^2}{2\beta_1\underline C-a_{\max}}
+\frac{\bigl[\beta_2(\overline C+x_{2'}^*)\bigr]^2}{2\beta_2\underline C-a_{\max}}.
\]
Assumption \textup{(36)} reads $S<\frac{2r}{K}\underline C$. Fix, once for all, as
\[
\theta\in\left(\frac{S}{\frac{2r}{K}\underline C},\,1\right).
\]
Then $\gamma_C(\theta)=\frac{2r}{K}\underline C-\frac{S}{\theta}>0$. Then all coefficients in (51) are strictly positive. Setting
\[
\beta:=\min\{\gamma_C(\theta),\eta,\lambda_1(\theta),\lambda_2(\theta)\}>0,
\]
and using
\[
V_N(t)=\tilde C(t)^2+\tilde x_{1'}(t)^2+\tilde x_{2'}(t)^2
+\sum_{\substack{i\in\mathcal T_N\\ i\in\mathcal{S}}}\tilde x_i(t)^2,
\]
we conclude from (51) that
\begin{equation}\label{equation: G15}
\dot V_N(t)\le -\beta\,V_N(t)\qquad \forall t\in[0,\infty).
\end{equation}

\medskip

By Gr\"onwall's lemma applied to (52),
\[
V_N(t)\le V_N(0)\,e^{-\beta t}\qquad\forall t\ge 0,\ \forall N.
\]

\medskip

Fix $t\ge 0$. Recall that $\mathcal T_N$ is the set of nodes up to depth $N$ (including $1',2'$), hence
\[
\mathcal T_N\subset \mathcal T_{N+1},\qquad \bigcup_{N\ge 0}\mathcal T_N=\mathcal S\cup\{1', 2'\},
\]
and every sum in $V_N(t)$ is finite. Define
\[
S_N(t):=\sum_{\substack{i\in\mathcal T_N\\ i\in\mathcal{S}}}\tilde x_i(t)^2,
\qquad
S(t):=\sum_{i\in\mathcal{S}}\tilde x_i(t)^2\in[0,\infty].
\]
Since all terms $\tilde x_i(t)^2\ge 0$, the sequence $(S_N(t))_{N\ge 0}$ is monotone nondecreasing:
\[
S_N(t)\le S_{N+1}(t)\qquad\text{for all }N,
\]
and by definition of $S(t)$ as the (possibly divergent) series over all indices,
\[
S_N(t)\uparrow S(t)\quad\text{as }N\to\infty.
\]
Therefore, for each fixed $t\ge 0$,
\[
V_N(t)
=\tilde C(t)^2+\tilde x_{1'}(t)^2+\tilde x_{2'}(t)^2 + S_N(t)
\uparrow
\tilde C(t)^2+\tilde x_{1'}(t)^2+\tilde x_{2'}(t)^2 + S(t)
=:V(t).
\]
The same argument at $t=0$ gives $V_N(0)\uparrow V(0)$. As part of the hypotheses, the initial deviation belongs to $\ell^2$, i.e.
\[
V(0)=\tilde C(0)^2+\tilde x_{1'}(0)^2+\tilde x_{2'}(0)^2+\sum_{i\in\mathcal{S}}\tilde x_i(0)^2<\infty.
\]
Then in particular $V_N(0)\le V(0)<\infty$ for every $N$. Now fix $t\ge 0$ and use the estimate obtained for each finite truncation:
\[
V_N(t)\le V_N(0)e^{-\beta t}\qquad\forall N.
\]
Since $V_N(0)\le V(0)$, we have the uniform bound
\[
V_N(t)\le V(0)e^{-\beta t}\qquad\forall N.
\]
Taking $N\to\infty$ and using the monotone convergence $V_N(t)\uparrow V(t)$ yields
\[
V(t)=\lim_{N\to\infty}V_N(t)
\le \lim_{N\to\infty} V_N(0)e^{-\beta t}
=\bigl(\lim_{N\to\infty}V_N(0)\bigr)e^{-\beta t}
=V(0)e^{-\beta t}.
\]
Hence
\[
V(t)\le V(0)e^{-\beta t}\qquad\forall t\ge 0.
\]

In particular $V(t)\to 0$ as $t\to\infty$, hence every component converges to its equilibrium value
with exponential rate $\beta$ (in the $\ell^2$ sense).
\end{proof}

Theorem~\ref{thm:core-GAS} states that, for the pure polymerization--competition model \eqref{eq:5}, the positive equilibrium associated with the upper stable resource branch \(C^*=C_+^*\) is globally exponentially stable in the stated \(\ell^2\)-framework, provided the persistence, dominance, and square-summability hypotheses of the theorem hold. We now extend the model by incorporating \emph{template-directed replication and selection} among the non-activated sequences. In particular, we keep the same core subsystem for the activated monomers and the competitor, while adding to each downstream variable $x_i$ a standard replicator term that accounts for differential fitness. This leads to the following binary polymerization system with competition and replication.

\section{A binary system of polymerization with competition and replication}
Taking in replication and competition into dynamics of the system of polymerization yields the following differential equation:
\begin{equation}\label{eq:7}
 \begin{aligned}
        \dot x_{1'}&=1-\beta_1 Cx_{1'},\\
        \dot x_{2'}&=1-\beta_2 Cx_{2'},\\
        \dot C&=C[r(1-\dfrac{C}{K})-\beta_1 x_{1'}-\beta_2 x_{2'}], \\
         \dot x_{i}&= a_{i}x_{i^{'}}-(d+ a_{i1}+a_{i2})x_{i}+\mathcal{R}x_i(f_i-\phi) \text{ for $i\in\mathcal{S}$}.
\end{aligned}
\end{equation}

The first three lines of this system describe the same as in Eq. (\ref{eq:5}). The first part of $\dot x_{i}$ in (\ref{eq:7}) pictures the dynamics of non-activated monomers $i$ as in (\ref{eq:5}), while the second part represents the standard selection equation. The parameter, $\mathcal{R},$ measures the relative rate of template-directed replication. $f_i$ denotes the fitness of sequence $i$. The amount $\phi$ is an additional decay rate, which represents the average fitness of the population, i.e.,

$$\phi=\dfrac{\Sigma_if_ix_i}{\Sigma_ix_i}\text{ for }i\in\mathcal{S}.$$

We denote the net reproductive rate of sequence $i$ (different from the activated monomers) by

\begin{equation}\label{eq:8}
g_i:= \mathcal{R}(f_i-\phi)-(d+a_{i1}+a_{i2}).
\end{equation}

Moreover, we consider $x=(x_i)_{i\in\mathcal{S}} \in\ell^1,$ and $0<f_i<1.$ Here is the formal proposition summarizing the nonnegativity and boundedness results for the system with replication dynamics:

\begin{proposition}[Persistence and destruction of the pre-replicative equilibrium under replication]
\label{lem:replication_preserve_loss}
Consider the polymerization--competition system without replication \eqref{eq:5}, and let
\[
E^*=(C^*,x_{1'}^*,x_{2'}^*,(x_i^*)_{i\in\mathcal S})
\]
be a strictly positive equilibrium of that system. Assume that the polymer equilibrium profile is summable and that the fitness landscape is bounded, namely
\[
(x_i^*)_{i\in\mathcal S}\in \ell^1_+(\mathcal S),
\qquad
(f_i)_{i\in\mathcal S}\in \ell^\infty(\mathcal S).
\]
Then the equilibrium fitness average
\begin{equation}\label{eq:phi_star_def}
\phi^*
:=
\frac{\sum_{i\in\mathcal S} f_i x_i^*}
{\sum_{i\in\mathcal S} x_i^*}
\end{equation}
is well-defined.

Consider now the replicated system \eqref{eq:7}, with replication intensity \(\mathcal R\ge 0\), the same coefficients
\[
\beta_1,\beta_2,r,K,\qquad a_i,\qquad k_i:=d+a_{i1}+a_{i2},
\]
and the same fitness landscape \((f_i)_{i\in\mathcal S}\). Then the following alternatives hold.

\begin{enumerate}
\item[\textnormal{(i)}] \textnormal{\bf Neutral preservation.}
If \(\mathcal R=0\), then \eqref{eq:7} coincides with \eqref{eq:5}. Hence \(E^*\) is an equilibrium of \eqref{eq:7}, with exactly the same stability properties as in the non-replicated system.

If \(\mathcal R>0\) and \(f_i\equiv f_\star\) is constant on \(\mathcal S\), then, for every admissible state with \(\sum_{i\in\mathcal S}x_i>0\), one has
\[
\phi=f_\star.
\]
Consequently, the replication term vanishes identically:
\[
\mathcal R x_i(f_i-\phi)=0
\qquad
(i\in\mathcal S).
\]
Thus the replicated system \eqref{eq:7} again coincides with the non-replicated system \eqref{eq:5}. In particular, \(E^*\) remains an equilibrium and all its stability properties are preserved.

\item[\textnormal{(ii)}] \textnormal{\bf Destruction under heterogeneous fitness.}
Assume \(\mathcal R>0\). Then \(E^*\) is an equilibrium of the replicated system \eqref{eq:7} if and only if
\[
f_i=\phi^*
\qquad
\text{for every } i\in\mathcal S.
\]
Since \(x_i^*>0\) for all \(i\), this is equivalent to saying that the fitness landscape is constant on \(\mathcal S\). Therefore, if the fitness landscape is heterogeneous, i.e. if there exist \(i,j\in\mathcal S\) with \(f_i\neq f_j\), then \(E^*\) is not an equilibrium of \eqref{eq:7}. In particular, \(E^*\) cannot be a Lyapunov stable equilibrium of the replicated system.
\end{enumerate}
\end{proposition}

\begin{proof}
If \(\mathcal R=0\), then \eqref{eq:7} reduces exactly to \eqref{eq:5}, so \(E^*\) is an equilibrium of \eqref{eq:7} and all stability properties are identical.

Assume next that \(\mathcal R>0\) and \(f_i\equiv f_\star\). Then, for every admissible state with \(\sum_i x_i>0\),
\[
\phi
=
\frac{\sum_{i\in\mathcal S} f_\star x_i}
{\sum_{i\in\mathcal S}x_i}
=
f_\star.
\]
Therefore \(f_i-\phi=0\) for every \(i\in\mathcal S\), and the replication term vanishes identically. Hence \eqref{eq:7} coincides with \eqref{eq:5}, proving the preservation claim.

Now assume \(\mathcal R>0\). Since \(E^*\) is an equilibrium of the non-replicated system, we have
\[
a_i x_{i'}^* - k_i x_i^*=0
\qquad
(i\in\mathcal S).
\]
If \(E^*\) is also an equilibrium of the replicated system \eqref{eq:7}, then for every \(i\in\mathcal S\),
\[
0
=
a_i x_{i'}^*-k_i x_i^*
+
\mathcal R x_i^*(f_i-\phi^*)
=
\mathcal R x_i^*(f_i-\phi^*).
\]
Because \(\mathcal R>0\) and \(x_i^*>0\), it follows that
\[
f_i=\phi^*
\qquad
\text{for every } i\in\mathcal S.
\]
Thus the fitness landscape must be constant. Conversely, if \(f_i=\phi^*\) for every \(i\), then the replication term vanishes at \(E^*\), and \(E^*\) satisfies the equilibrium equations of \eqref{eq:7}. This proves the necessary and sufficient condition.

Hence, if the fitness landscape is heterogeneous, \(E^*\) cannot be an equilibrium of \eqref{eq:7}. Since Lyapunov stability as an equilibrium presupposes that the point is an equilibrium, \(E^*\) cannot be a Lyapunov stable equilibrium of the replicated system.
\end{proof}

\begin{lemma}[Pairwise log-ratio identity]
\label{lem:log_ratio_identity}
Consider the replicated system \eqref{eq:7}. Let \(i,j\in\mathcal S\), and suppose that along a positive trajectory one has
\[
x_i(t)>0,\qquad x_j(t)>0
\qquad
\text{for all } t\ge 0.
\]
Assume also that the average fitness \(\phi(t)\) is well-defined along the trajectory. Then
\begin{equation}\label{eq:log_ratio_identity}
\frac{d}{dt}\log\frac{x_i(t)}{x_j(t)}
=
\left(
a_i\frac{x_{i'}(t)}{x_i(t)}
-
a_j\frac{x_{j'}(t)}{x_j(t)}
\right)
-(k_i-k_j)
+
\mathcal R(f_i-f_j).
\end{equation}
In particular, the average fitness term \(\phi(t)\) cancels from all pairwise log-ratio dynamics.
\end{lemma}

\begin{proof}
For each \(q\in\mathcal S\), the replicated equation is
\[
\dot x_q
=
a_q x_{q'}-k_q x_q+\mathcal R x_q(f_q-\phi).
\]
For \(x_q(t)>0\), division by \(x_q(t)\) gives
\[
\frac{\dot x_q}{x_q}
=
a_q\frac{x_{q'}}{x_q}
-k_q
+\mathcal R(f_q-\phi).
\]
Applying this identity to \(q=i\) and \(q=j\), and subtracting, we get
\[
\frac{d}{dt}\log\frac{x_i}{x_j}
=
\frac{\dot x_i}{x_i}
-
\frac{\dot x_j}{x_j}
=
\left(
a_i\frac{x_{i'}}{x_i}
-
a_j\frac{x_{j'}}{x_j}
\right)
-(k_i-k_j)
+
\mathcal R(f_i-f_j),
\]
because the two occurrences of \(\phi\) cancel.
\end{proof}

\begin{proposition}[Pairwise exponential divergence under replication dominance]
\label{lem:ratio_divergence}
Consider the replicated system \eqref{eq:7}. Fix two indices \(i,j\in\mathcal S\), and suppose that along a positive trajectory
\[
x_i(t)>0,\qquad x_j(t)>0
\qquad
\text{for all } t\ge 0.
\]
Assume that the average fitness \(\phi(t)\) is well-defined along the trajectory. Define the fitness gap
\[
\Delta_{ij}:=f_i-f_j,
\]
and assume
\[
\Delta_{ij}>0.
\]
Suppose moreover that the kinetic forcing into the less fit sequence \(j\) is uniformly bounded along the trajectory, namely that there exists \(M_j<\infty\) such that
\begin{equation}\label{eq:Mj_bound}
0\le
a_j\frac{x_{j'}(t)}{x_j(t)}
\le M_j
\qquad
\text{for all } t\ge 0.
\end{equation}
Set
\[
\kappa_{ij}^+:=(k_i-k_j)_+
:=
\max\{k_i-k_j,0\}.
\]
If
\begin{equation}\label{eq:R_threshold}
\mathcal R\Delta_{ij}
>
M_j+\kappa_{ij}^+,
\end{equation}
then there exists
\[
\varepsilon
:=
\mathcal R\Delta_{ij}-M_j-\kappa_{ij}^+
>0
\]
such that
\[
\frac{d}{dt}\log\frac{x_i(t)}{x_j(t)}
\ge
\varepsilon
\qquad
\text{for all } t\ge 0.
\]
Consequently,
\begin{equation}\label{eq:pairwise_exp_divergence}
\frac{x_i(t)}{x_j(t)}
\ge
\frac{x_i(0)}{x_j(0)}e^{\varepsilon t}
\qquad
\text{for all } t\ge 0,
\end{equation}
and hence
\[
\lim_{t\to\infty}\frac{x_i(t)}{x_j(t)}
=
+\infty.
\]
\end{proposition}

\begin{proof}
By Lemma~\ref{lem:log_ratio_identity},
\[
\frac{d}{dt}\log\frac{x_i}{x_j}
=
\left(
a_i\frac{x_{i'}}{x_i}
-
a_j\frac{x_{j'}}{x_j}
\right)
-(k_i-k_j)
+
\mathcal R\Delta_{ij}.
\]
Since \(a_i>0\), \(x_{i'}(t)\ge 0\), and \(x_i(t)>0\), we have
\[
a_i\frac{x_{i'}(t)}{x_i(t)}\ge 0.
\]
Moreover, by \eqref{eq:Mj_bound},
\[
-a_j\frac{x_{j'}(t)}{x_j(t)}
\ge
-M_j.
\]
Finally,
\[
-(k_i-k_j)\ge -(k_i-k_j)_+=-\kappa_{ij}^+.
\]
Combining these estimates gives
\[
\frac{d}{dt}\log\frac{x_i}{x_j}
\ge
-M_j-\kappa_{ij}^+
+
\mathcal R\Delta_{ij}
=
\varepsilon>0.
\]
Integrating over \([0,t]\), we get

\[
\log\frac{x_i(t)}{x_j(t)}
-
\log\frac{x_i(0)}{x_j(0)}
\ge
\varepsilon t.
\]
Exponentiating yields \eqref{eq:pairwise_exp_divergence}, and therefore
\[
\frac{x_i(t)}{x_j(t)}\to+\infty
\qquad
\text{as } t\to\infty.
\]
\end{proof}

\begin{remark}
A simpler but slightly more conservative sufficient condition is
\[
\mathcal R\Delta_{ij}>M_j+|k_i-k_j|.
\]
This follows from the estimate
\[
-(k_i-k_j)\ge -|k_i-k_j|.
\]
The sharper condition \eqref{eq:R_threshold} uses only the positive part \((k_i-k_j)_+\), because a larger decay rate for \(j\) relative to \(i\) already helps the divergence of \(x_i/x_j\).
\end{remark}

\begin{proposition}[Uniform divergence along two full lineages]
\label{lem:lineage_divergence}
Consider the replicated system \eqref{eq:7}, and assume that the solution is strictly positive, i.e.
\[
x_i(t)>0
\qquad
\text{for all } i\in\mathcal S,\ t\ge 0.
\]
Assume also that the average fitness \(\phi(t)\) is well-defined along the trajectory.

Let
\[
\mathcal L=\{\ell_n\}_{n\ge 0},
\qquad
\mathcal M=\{m_n\}_{n\ge 0}
\]
be two full lineages, meaning that
\[
\ell_0,m_0\in\{1',2'\},
\qquad
\ell_n'= \ell_{n-1},
\qquad
m_n'=m_{n-1}
\qquad
(n\ge 1).
\]
Assume that the kinetic forcing into the second lineage is uniformly bounded: there exists \(M_{\mathcal M}<\infty\) such that, for all \(n\ge 1\) and all \(t\ge 0\),
\begin{equation}\label{eq:forcing_bound_lineage_M}
0\le
a_{m_n}\frac{x_{m_{n-1}}(t)}{x_{m_n}(t)}
\le
M_{\mathcal M}.
\end{equation}
Assume also that the relative decay disadvantage of the first lineage is uniformly bounded:
\begin{equation}\label{eq:kappa_plus_def}
\kappa_+
:=
\sup_{n\ge 1}(k_{\ell_n}-k_{m_n})_+
<\infty,
\qquad
k_q:=d+a_{q1}+a_{q2}.
\end{equation}
Finally, assume that the first lineage has a uniform fitness advantage over the second:
\begin{equation}\label{eq:Delta_def}
\Delta
:=
\inf_{n\ge 1}\bigl(f_{\ell_n}-f_{m_n}\bigr)
>0.
\end{equation}

If the replication intensity satisfies
\begin{equation}\label{eq:R_threshold_lineage}
\mathcal R\Delta
>
M_{\mathcal M}+\kappa_+,
\end{equation}
then there exists
\[
\varepsilon
:=
\mathcal R\Delta-M_{\mathcal M}-\kappa_+
>0
\]
such that, for every \(n\ge 1\) and every \(t\ge 0\),
\[
\frac{d}{dt}\log\frac{x_{\ell_n}(t)}{x_{m_n}(t)}
\ge
\varepsilon.
\]
Consequently, for every \(n\ge 1\),
\begin{equation}\label{eq:lineage_exp_divergence}
\frac{x_{\ell_n}(t)}{x_{m_n}(t)}
\ge
\frac{x_{\ell_n}(0)}{x_{m_n}(0)}e^{\varepsilon t}
\qquad
\text{for all } t\ge 0.
\end{equation}
In particular,
\[
\lim_{t\to\infty}
\frac{x_{\ell_n}(t)}{x_{m_n}(t)}
=
+\infty
\qquad
\text{for every fixed } n\ge 1.
\]
Thus the first lineage separates exponentially from the second at every depth, with a rate \(\varepsilon\) independent of \(n\).
\end{proposition}

\begin{proof}
Fix \(n\ge 1\). Applying Lemma~\ref{lem:log_ratio_identity} with
\[
i=\ell_n,
\qquad
j=m_n,
\]
we get
\[
\frac{d}{dt}\log\frac{x_{\ell_n}}{x_{m_n}}
=
\left(
a_{\ell_n}\frac{x_{\ell_{n-1}}}{x_{\ell_n}}
-
a_{m_n}\frac{x_{m_{n-1}}}{x_{m_n}}
\right)
-(k_{\ell_n}-k_{m_n})
+
\mathcal R(f_{\ell_n}-f_{m_n}).
\]
Since
\[
a_{\ell_n}\frac{x_{\ell_{n-1}}}{x_{\ell_n}}\ge 0,
\]
and by \eqref{eq:forcing_bound_lineage_M},
\[
-a_{m_n}\frac{x_{m_{n-1}}}{x_{m_n}}
\ge
-M_{\mathcal M},
\]
while by \eqref{eq:kappa_plus_def},
\[
-(k_{\ell_n}-k_{m_n})
\ge
-(k_{\ell_n}-k_{m_n})_+
\ge
-\kappa_+.
\]
Finally, by \eqref{eq:Delta_def},
\[
f_{\ell_n}-f_{m_n}\ge \Delta.
\]
Therefore, for all \(n\ge 1\) and all \(t\ge 0\),
\[
\frac{d}{dt}\log\frac{x_{\ell_n}(t)}{x_{m_n}(t)}
\ge
-M_{\mathcal M}-\kappa_+
+
\mathcal R\Delta
=
\varepsilon>0.
\]
Integrating over \([0,t]\) gives
\[
\log\frac{x_{\ell_n}(t)}{x_{m_n}(t)}
-
\log\frac{x_{\ell_n}(0)}{x_{m_n}(0)}
\ge
\varepsilon t.
\]
Exponentiating yields \eqref{eq:lineage_exp_divergence}. Hence, for every fixed \(n\ge 1\),
\[
\frac{x_{\ell_n}(t)}{x_{m_n}(t)}\to+\infty
\qquad
\text{as }t\to\infty.
\]
The rate \(\varepsilon\) is independent of \(n\), although the prefactor
\[
\frac{x_{\ell_n}(0)}{x_{m_n}(0)}
\]
may depend on \(n\).
\end{proof}

\begin{remark}
If one wants a divergence estimate that is also uniform in the prefactor over all depths, it is enough to assume
\[
\inf_{n\ge 1}
\frac{x_{\ell_n}(0)}{x_{m_n}(0)}
>0.
\]
Under this additional condition, \eqref{eq:lineage_exp_divergence} implies a depth-uniform lower bound of the form
\[
\frac{x_{\ell_n}(t)}{x_{m_n}(t)}
\ge
c_0 e^{\varepsilon t}
\qquad
\text{for all } n\ge 1,\ t\ge 0,
\]
where
\[
c_0:=
\inf_{n\ge 1}
\frac{x_{\ell_n}(0)}{x_{m_n}(0)}
>0.
\]
\end{remark}

\section{Numerical results}
\label{sec:numerics}

We now present numerical illustrations of the analytical results. Following the numerical style of prevolutionary dynamics in \cite{Nowak2008}, we display equilibrium abundances as functions of biologically meaningful control parameters: the kinetic selection intensity \(s\), the compound resource parameter \(rK\), and the replication intensity \(\mathcal R\). Since the full polymerization network is infinite-dimensional, all simulations are performed on finite binary truncations
\[
\mathcal S_N=\bigcup_{n=1}^{N}\{1,2\}^n.
\]
Unless otherwise stated, we use \(N=6\). The truncation contains
\[
|\mathcal S_N|=\sum_{n=1}^{N}2^n=2^{N+1}-2
\]
polymer variables.

\subsection{Kinetic selection in the pre-replicative polymerization landscape}

We first illustrate how kinetic asymmetries alone can create unequal equilibrium abundances before replication is introduced. For each sequence \(i\), the equilibrium abundance is
\[
x_i^*
=
\prod_{r=1}^{L(i)}
\frac{a_{i^{(r)}}}{d+a_{i^{(r)}1}+a_{i^{(r)}2}}.
\]
Thus the polymerization rates define a pre-replicative abundance landscape.

In the random landscape experiment, we introduce kinetic heterogeneity by assigning each polymerization reaction one of two possible rates. More precisely, for every reaction leading to a sequence \(q\), we set
\[
a_q(s)=1+s\eta_q,
\]
where \(\eta_q\in\{0,1\}\). A fixed fraction \(p\) of the reactions is chosen to be fast, meaning \(\eta_q=1\), while the remaining reactions are slow, meaning \(\eta_q=0\). Thus fast reactions have rate \(1+s\), whereas slow reactions have rate \(1\). The same random assignment of fast and slow reactions is kept fixed as \(s\) varies.

The parameter \(s\ge0\) therefore measures the strength of the kinetic contrast between fast and slow reactions. When \(s=0\), all reactions have the same rate and the polymerization landscape is kinetically homogeneous. As \(s\) increases, fast reactions become progressively more favorable relative to slow ones, so that lineages containing several fast reactions along their ancestral paths can become strongly enriched.

For each value of \(s\), we compute the equilibrium \(x_i^*(s)\) on the finite truncation \(\mathcal S_N\). Since
\[
x_i^*(s)
=
\prod_{r=1}^{L(i)}
\frac{a_{i^{(r)}}(s)}
{d+a_{i^{(r)}1}(s)+a_{i^{(r)}2}(s)},
\]
the effect of \(s\) is cumulative along the ancestral chain of \(i\). A sequence whose lineage contains many fast incoming reactions may gain a multiplicative advantage, whereas fast outgoing reactions from an intermediate sequence may also increase its loss term through the denominator. Hence the abundance landscape depends on the balance between enhanced production into a sequence and enhanced depletion toward its followers.

Figure~\ref{fig:random-prelife-landscape} plots all equilibrium abundances as functions of \(s\). As \(s\) increases, the initially homogeneous landscape becomes increasingly structured: a small subset of lineages becomes highly enriched, while other lineages are suppressed. This illustrates a selection-like effect generated purely by kinetic asymmetries, before template-directed replication is introduced.

\begin{figure}[h]
\centering
\includegraphics[width=0.75\textwidth]{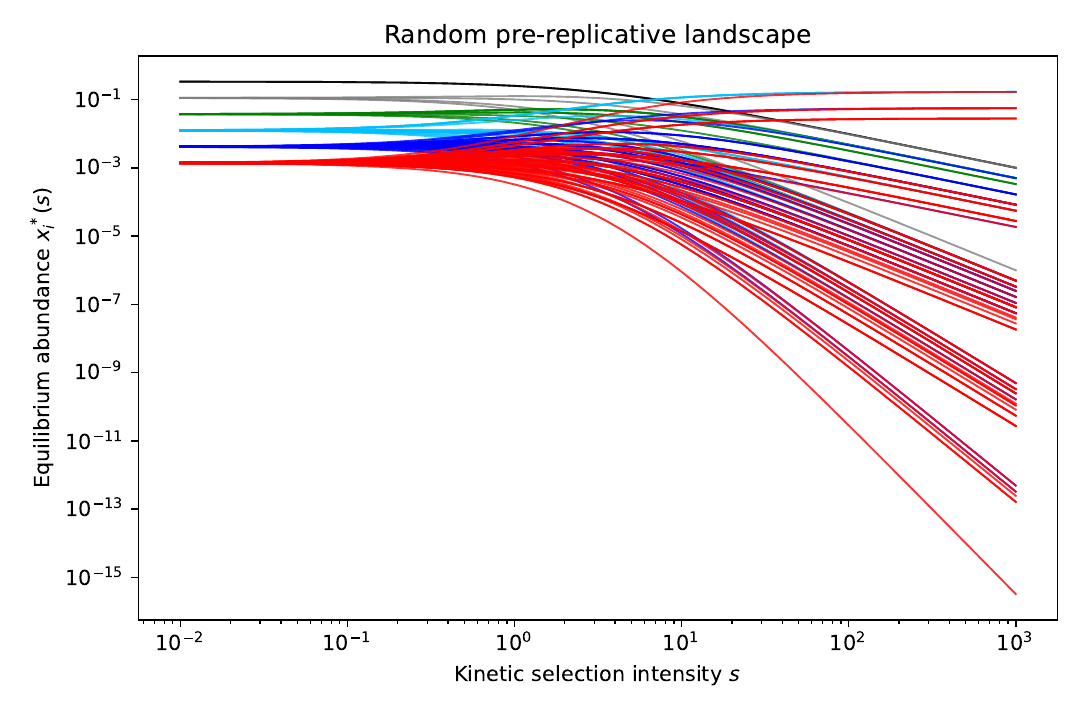}
\caption{Random kinetic landscape. Equilibrium abundances \(x_i^*(s)\) for all sequences of length \(1\le L(i)\le 6\), plotted as functions of the kinetic selection intensity \(s\). A fraction \(p=1/2\) of reactions has rate \(1+s\), while the remaining reactions have rate \(1\).}
\label{fig:random-prelife-landscape}
\end{figure}

We also consider a single-lineage landscape in which one prescribed master sequence \(i^\star\) of length \(N\) receives increased production rates along its entire ancestral path. Specifically, all reactions along the lineage leading to \(i^\star\) have rate \(1+s\), while all other reactions have rate \(1\). Figure~\ref{fig:master-landscape} shows that, for small \(s\), the master sequence is not strongly distinguished from the background. As \(s\) grows, however, the master lineage becomes increasingly dominant among sequences of the same length.

\begin{figure}[h]
\centering
\includegraphics[width=0.75\textwidth]{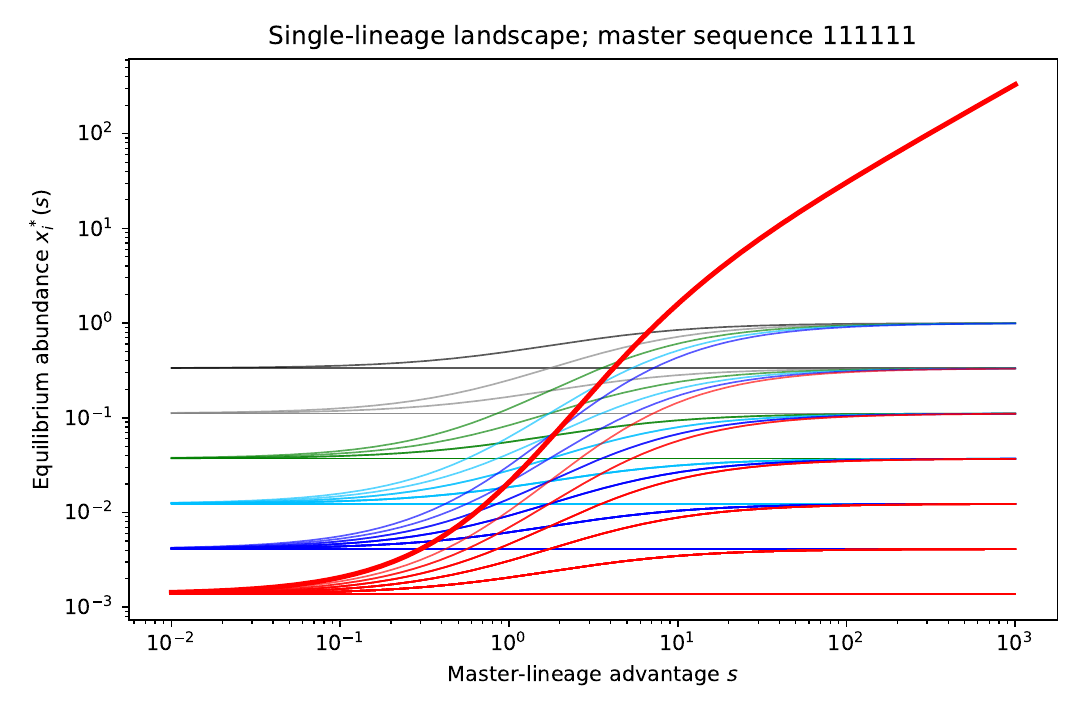}
\caption{Single-lineage kinetic landscape. The master sequence \(i^\star=111111\) is favored by assigning rate \(1+s\) to all reactions along its ancestral path. The master lineage becomes increasingly enriched as \(s\) increases.}
\label{fig:master-landscape}
\end{figure}

\subsection{Resource threshold and bifurcation of positive equilibria}

For the coupled polymerization with competition system, positive equilibria exist if and only if
\[
rK\ge 8.
\]
When \(rK>8\), the two positive resource equilibria are
\[
C_\pm^*
=
\frac{K}{2}
\left(
1\pm \sqrt{1-\frac{8}{rK}}
\right).
\]
The reduced \(C\)-dynamics has linearization coefficient
\[
f'(C^*)=\frac{r}{K}(K-2C^*).
\]
Hence \(C_-^*<K/2\) is unstable, while \(C_+^*>K/2\) is locally stable. Figure~\ref{fig:resource-bifurcation} shows the two branches as functions of \(rK\). The critical value \(rK=8\) marks the saddle-node type threshold at which positive resource equilibria appear.

\begin{figure}[h]
\centering
\includegraphics[width=0.7\textwidth]{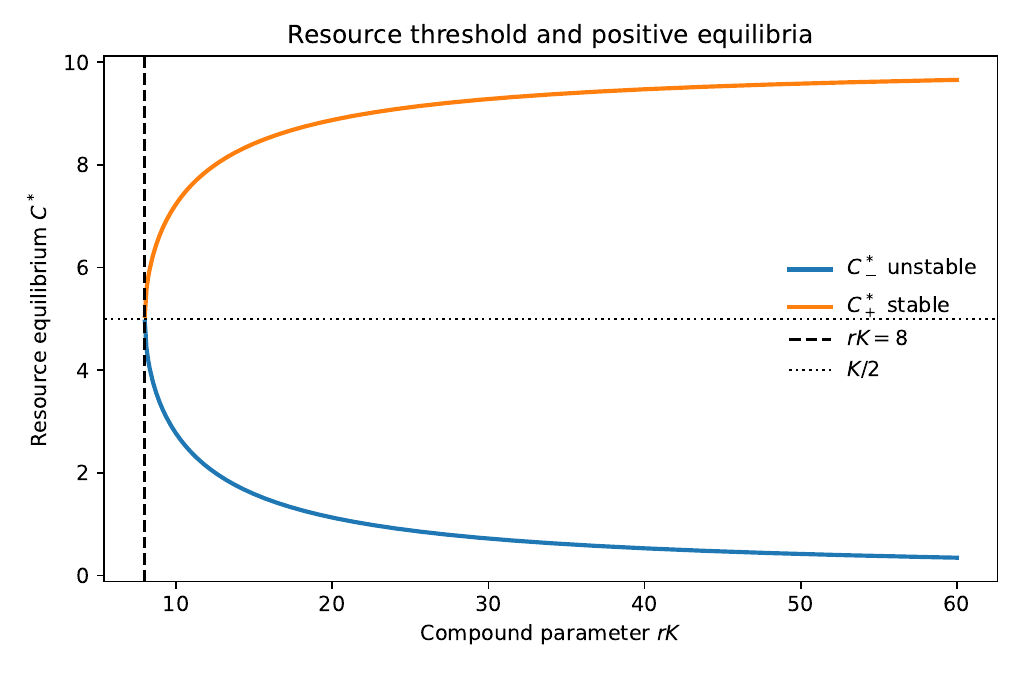}
\caption{Resource bifurcation diagram. Positive equilibria exist only when \(rK\ge 8\). For \(rK>8\), the lower branch \(C_-^*\) is unstable and the upper branch \(C_+^*\) is stable for the reduced \(C\)-dynamics.}
\label{fig:resource-bifurcation}
\end{figure}

For example, if \(K=10\) and \(r=1.2\), then \(rK=12>8\), and
\[
C_-^*\approx 2.113,
\qquad
C_+^*\approx 7.887.
\]
Taking \(\beta_1=1\) and \(\beta_2=1.3\), the associated activated precursor equilibria on the stable branch are
\[
x_{1'}^*=\frac{1}{\beta_1C_+^*}\approx 0.127,
\qquad
x_{2'}^*=\frac{1}{\beta_2C_+^*}\approx 0.098.
\]

\subsection{Exponential convergence of the coupled truncated system}

We next simulate the finite truncation of the polymerization--competition model:
\[
\begin{aligned}
\dot x_{1'}&=1-\beta_1Cx_{1'},\\
\dot x_{2'}&=1-\beta_2Cx_{2'},\\
\dot C&=C\left[r\left(1-\frac{C}{K}\right)-\beta_1x_{1'}-\beta_2x_{2'}\right],\\
\dot x_i&=a_ix_{i'}-\left(d+a_{i1}+a_{i2}\right)x_i,
\qquad i\in\mathcal S_N.
\end{aligned}
\]
We use
\[
d=1,\qquad a_i=0.2,\qquad a_{i1}=a_{i2}=0.1,
\qquad K=10,\qquad r=1.2,\qquad \beta_1=1,\qquad \beta_2=1.3.
\]
The stable equilibrium is determined by \(C^*=C_+^*\), together with
\[
x_{1'}^*=\frac{1}{\beta_1C^*},
\qquad
x_{2'}^*=\frac{1}{\beta_2C^*},
\qquad
x_i^*=
\frac{a_i}{d+a_{i1}+a_{i2}}x_{i'}^*.
\]
To quantify convergence, define the truncated Lyapunov error
\[
V_N(t)
=
|C(t)-C^*|^2
+
|x_{1'}(t)-x_{1'}^*|^2
+
|x_{2'}(t)-x_{2'}^*|^2
+
\sum_{i\in\mathcal S_N}|x_i(t)-x_i^*|^2.
\]
Figure~\ref{fig:coupled-convergence} shows exponential decay of \(V_N(t)\), consistent with the Lyapunov estimate proved for the infinite system under the stated dominance assumptions.

\begin{figure}[h]
\centering
\includegraphics[width=0.72\textwidth]{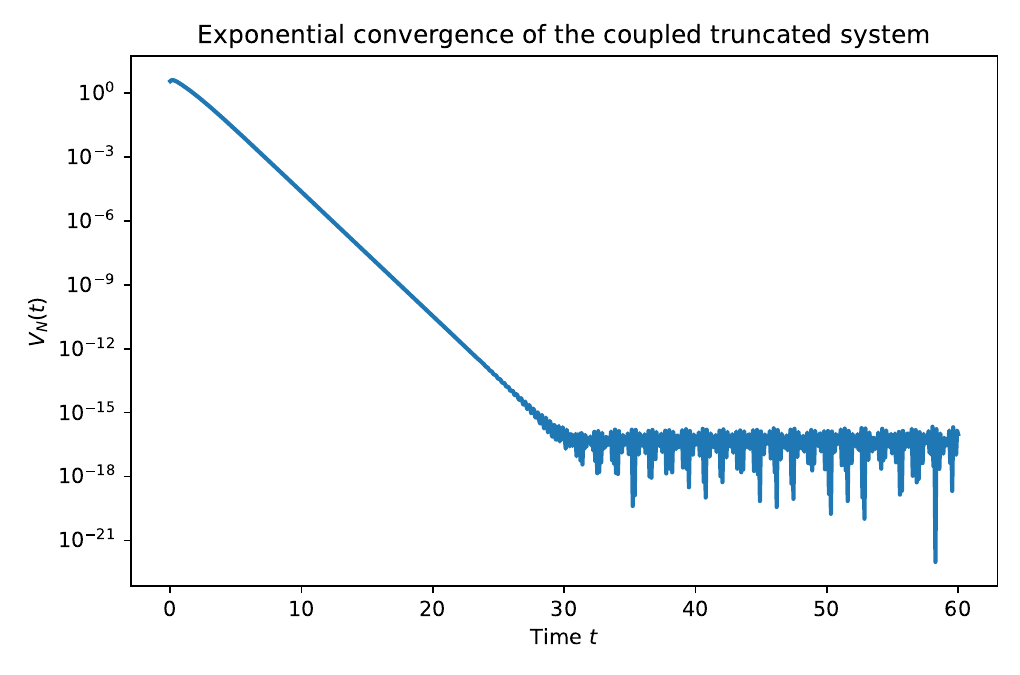}
\caption{Convergence of the truncated polymerization--competition system toward the stable equilibrium associated with \(C_+^*\). The semilogarithmic plot of \(V_N(t)\) shows exponential decay.}
\label{fig:coupled-convergence}
\end{figure}

\subsection{Competition between prelife and replication}

Finally, we illustrate the transition produced by template-directed replication. We use the replicated finite system
\[
\dot x_i
=
a_ix_{i'}-\left(d+a_{i1}+a_{i2}\right)x_i
+
\mathcal R x_i(f_i-\phi),
\qquad
i\in\mathcal S_N,
\]
with
\[
\phi(t)
=
\frac{\sum_{i\in\mathcal S_N}f_ix_i(t)}
{\sum_{i\in\mathcal S_N}x_i(t)}.
\]
Only sequences of maximal length \(N\) are assigned nonzero heterogeneous fitness values; shorter sequences are assigned neutral background fitness. This mirrors the experiment in which long polymers can become replicators while shorter polymers remain part of the pre-replicative chemical background.

For each value of \(\mathcal R\), the system is integrated until it approaches a numerical steady state. Figure~\ref{fig:replication-sweep} plots the resulting abundances as functions of the replication intensity. For small \(\mathcal R\), the abundance structure is dominated by polymerization and longer sequences remain rare. As \(\mathcal R\) increases, the fittest replicating sequence becomes strongly enriched, producing a transition from a pre-replicative regime to a replication-dominated regime.

\begin{figure}[H]
\centering
\includegraphics[width=0.75\textwidth]{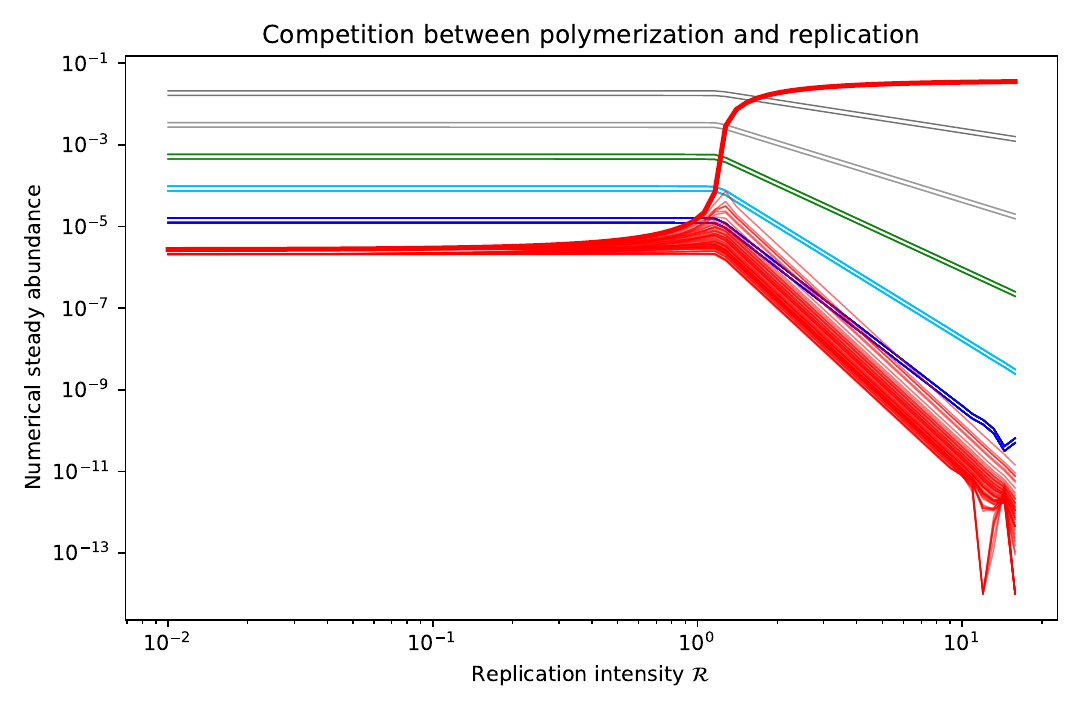}
\caption{Competition between pre-replicative polymerization and template-directed replication. Equilibrium abundances are plotted as functions of the replication intensity \(\mathcal R\). For sufficiently large \(\mathcal R\), the fittest replicator dominates.}
\label{fig:replication-sweep}
\end{figure}

A complementary way to visualize selection is to compare two sequences \(i\) and \(j\) with different fitness values. The replicated dynamics gives
\[
\frac{d}{dt}\log\frac{x_i}{x_j}
=
\left(a_i\frac{x_{i'}}{x_i}
-
a_j\frac{x_{j'}}{x_j}\right)
-
(k_i-k_j)
+
\mathcal R(f_i-f_j),
\]
because the average fitness \(\phi\) cancels. Therefore, if the replication advantage dominates the kinetic forcing and decay imbalance, the log-ratio grows approximately linearly. Figure~\ref{fig:log-ratio} shows this effect for \(i=111111\) and \(j=222222\), with \(f_i>f_j\). Increasing \(\mathcal R\) produces faster separation between the two lineages.

\begin{figure}[H]
\centering
\includegraphics[width=0.72\textwidth]{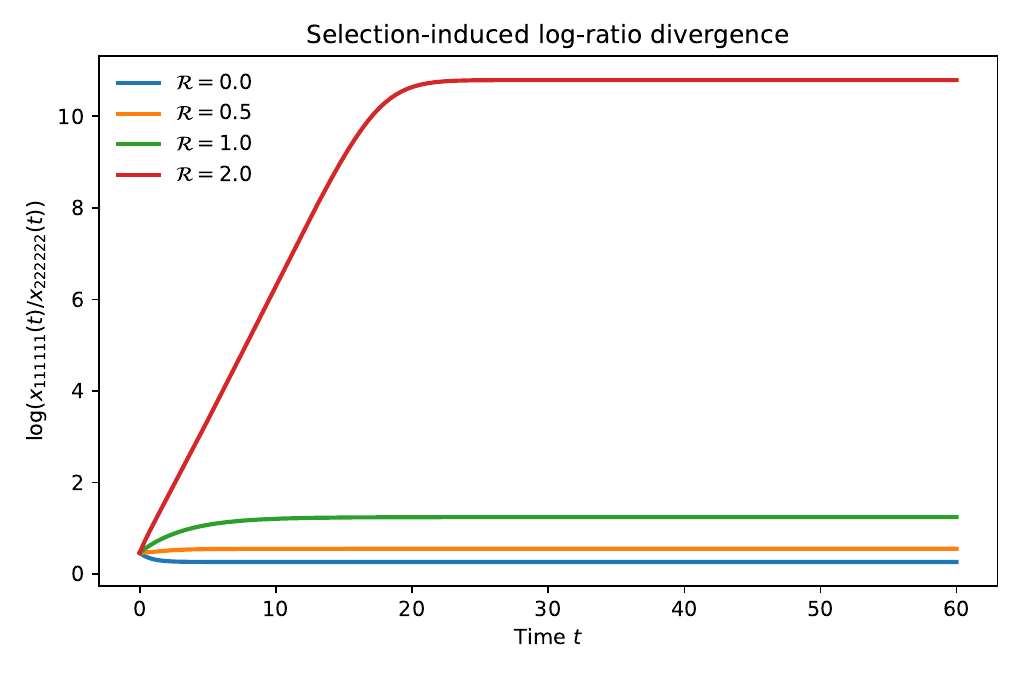}
\caption{Selection-induced divergence between two lineages. For sufficiently large \(\mathcal R\), the log-ratio \(\log(x_i(t)/x_j(t))\) grows approximately linearly, indicating exponential separation of abundances.}
\label{fig:log-ratio}
\end{figure}

These numerical experiments support the analytical picture developed above. The purely kinetic system converges to a unique compositional background; the resource-coupled system exhibits a sharp feasibility threshold \(rK=8\); and heterogeneous replication breaks the previous equilibrium by amplifying fitter lineages.

\FloatBarrier
\section{Concluding remarks}

In this paper, we developed a deterministic framework for sequence polymerization on the directed genealogy of finite molecules. The model separates three increasingly rich dynamical regimes, a purely kinetic polymerization cascade with fixed activated precursors, a polymerization--competition system, and a replicated system incorporating template-directed selection. This progression allows one to distinguish between compositional structure generated by kinetic production alone and genuine selection effects produced by differential replication.

For the baseline polymerization model, the infinite triangular structure of the equations makes it possible to construct solutions recursively along the polymerization tree. This yields global componentwise well-posedness, positivity, and a unique strictly positive equilibrium. The equilibrium has an explicit product representation along ancestor chains, showing that each sequence abundance is computed by the cumulative balance between production from its precursor and losses due to degradation and further extension. Moreover, every component converges exponentially to its equilibrium value. Under the additional uniform net decay condition
\[
\inf_{i\in\mathcal S}(k_i-a_i)>0,
\]
this convergence strengthens to global exponential stability in \(\ell^\infty(\mathcal S)\).

The binary model with resource reveals a second layer of structure. Once the activated precursor concentrations are allowed to evolve together with a shared environmental resource \(C\), the existence of strictly positive equilibria is governed by the scalar threshold condition
\[
rK\ge 8.
\]
At the critical value \(rK=8\), a unique positive resource level appears, whereas for \(rK>8\) two positive resource branches exist. The reduced scalar \(C\)-dynamics identifies the upper branch \(C_+^*\) as the locally attracting branch and the lower branch \(C_-^*\) as unstable.

We also established the basic dynamical consistency of the resource-coupled system. Positivity is forward invariant, global componentwise solutions exist for strictly positive initial data, and the resource variable admits an explicit logistic upper bound. A further lower-barrier condition for \(C(t)\) provides the persistence estimate needed to prevent degeneration of the resource level. Combining this persistence with uniform kinetic dominance assumptions gives a Lyapunov estimate and exponential convergence toward the positive equilibrium associated with the stable resource branch. In this way, the infinite polymerization cascade inherits stability from a controlled finite-dimensional core, provided that the coupling terms are dominated by sufficient dissipation.

The addition of replication changes the qualitative picture. When the fitness landscape is neutral, the replication term vanishes identically and the pre-replicative equilibrium is preserved. By contrast, heterogeneous fitness destroys the previous polymerization--competition equilibrium: unless all fitness values coincide with the equilibrium average fitness, the replicated vector field does not vanish at the old equilibrium. Thus differential replication is not just a perturbation of the kinetic background, it creates a genuinely new dynamical scenario.

The log-ratio identities make this transition especially transparent. For any two positive sequence concentrations \(x_i\) and \(x_j\), the average fitness term cancels from
\[
\frac{d}{dt}\log\frac{x_i}{x_j}.
\]
Then, relative growth is governed by three competing effects: kinetic forcing from precursors, differences in decay and extension losses, and the direct fitness gap multiplied by the replication intensity \(\mathcal R\). When the replication advantage dominates the kinetic and decay imbalance, the log-ratio grows at least linearly in time, and the concentration ratio diverges exponentially. The same mechanism extends to entire lineages under a uniform fitness gap, giving lineage-level separation at a rate independent of depth.

The numerical experiments support the mathematical analysis. In the pre-replicative scenario, kinetic asymmetries alone generate nonuniform abundance landscapes, including enrichment of favored lineages. In the resource-coupled system, the bifurcation diagram illustrates the threshold \(rK=8\) and the emergence of the stable upper resource branch. Simulations of finite truncations show exponential decay of the Lyapunov error toward the stable equilibrium. Finally, increasing the replication intensity produces a transition from a polymerization-dominated scenario to a replication-dominated scenario in which fitter sequences separate exponentially from their competitors.

Overall, the results suggest a coherent mathematical description for the emergence of evolutionary dynamics from open polymerization chemistry. A driven polymerization network first creates a stable compositional background; resource coupling imposes feasibility thresholds and selects admissible steady regimes; and heterogeneous replication breaks the pre-replicative equilibrium by amplifying fitter lineages. In this sense, selection appears as a dynamical instability of the kinetic background once replication becomes strong enough to overcome production and decay constraints.

Several directions are still open. One is to weaken the sufficient dominance assumptions used in the \(\ell^2\)-Lyapunov argument, or to replace them with sharper spectral or semigroup conditions. A further direction is to incorporate mutation, catalytic feedback, compartmentalization, or time-dependent environmental forcing. Such directions would bring the model closer to realistic prebiotic scenarios, where polymerization, degradation, resource limitation, and replication occur simultaneously in fluctuating far from equilibrium contexts.


\begin{thebibliography}{99}
\bibitem[Schopf(2024)]{Schopf24}
Schopf, J. W. (2024). \textit{Pioneers of origin of life studies--Darwin, Oparin, Haldane, Miller, Or\'{o}--and the oldest known records of life}.
\textit{Life}, 14(10), 1345.
https://doi.org/10.3390/life14101345

\bibitem[Kauffman(2011)]{kauffman2011approaches}
Kauffman, Stuart A (2011).
\textit{Approaches to the origin of life on earth}.
\textit{Life} 1(1), 34--48
https://doi.org/10.3390/life1010034

\bibitem[Agutter and Wheatley(2008)]{agutter2008thinking}
Agutter, Paul S \& Wheatley, Denys N (2008).
\textit{Thinking about life: the history and philosophy of biology and other sciences}.
\textit{Springer}, p.  224.
https://doi.org/10.1007/978-1-4020-8866-7


\bibitem[Eigen(1971)]{Eigen1971}
Eigen, M. (1971).
\textit{Selforganization of matter and the evolution of biological macromolecules}.
\textit{Naturwissenschaften}, 58, 465--523.
https://doi.org/10.1007/BF00623322

\bibitem[Eigen and Schuster(1977)]{EigenSchuster1977}
Eigen, M., \& Schuster, P. (1977).
\textit{A principle of natural self-organization. Part A: Emergence of the hypercycle}.
\textit{Naturwissenschaften}, 64, 541--565.
https://doi.org/10.1007/BF00450633

\bibitem[Nowak and Ohtsuki(2008)]{Nowak2008}
Nowak, M. A., \& Ohtsuki, H. (2008).
\textit{Prevolutionary dynamics and the origin of evolution}.
\textit{Proceedings of the National Academy of Sciences of the United States of America}, 105(39), 14924--14927.
https://doi.org/10.1073/pnas.0806714105

\bibitem[Manapat et al.(2009)]{Manapat2009}
Manapat, M., Ohtsuki, H., Bürger, R., \& Nowak, M. A. (2009).
\textit{Originator dynamics}.
\textit{Journal of Theoretical Biology}, 256(4), 586--595.
https://doi.org/10.1016/j.jtbi.2008.10.006

\bibitem[Sutherland(2017)]{Sutherland2017}
Sutherland, J. D. (2017).
\textit{Opinion: Studies on the origin of life --- the end of the beginning}.
\textit{Nature Reviews Chemistry}, 1, 0012.
https://doi.org/10.1038/s41570-016-0012

\bibitem[Pascal et al.(2013)]{Pascal2013}
Pascal, R., Pross, A., \& Sutherland, J. D. (2013).
\textit{Towards an evolutionary theory of the origin of life based on kinetics and thermodynamics}.
\textit{Open Biology}, 3, 130156.
https://doi.org/10.1098/rsob.130156

\bibitem[Baaske et al.(2007)]{Baaske2007}
Baaske, P., Weinert, F. M., Duhr, S., Lemke, K. H., Russell, M. J., \& Braun, D. (2007).
\textit{Extreme accumulation of nucleotides in simulated hydrothermal pore systems}.
\textit{Proceedings of the National Academy of Sciences of the United States of America}, 104(22), 9346--9351.
https://doi.org/10.1073/pnas.0609592104

\bibitem[Ianeselli et al.(2023)]{Ianeselli2023}
Ianeselli, A., Salditt, A., Mast, C., Ercolano, B., Kufner, C. L., Scheu, B., \& Braun, D. (2023).
\textit{Physical non-equilibria for prebiotic nucleic acid chemistry}.
\textit{Nature Reviews Physics}, 5, 185--195.
https://doi.org/10.1038/s42254-022-00550-3

\bibitem[Yi et al.(2020)]{Yi2020}
Yi, R., Tran, Q. P., Ali, S., Yoda, I., Adam, Z. R., Cleaves II, H. J., \& Fahrenbach, A. C. (2020).
\textit{A continuous reaction network that produces RNA precursors}.
\textit{Proceedings of the National Academy of Sciences of the United States of America}, 117(24), 13267--13274.
https://doi.org/10.1073/pnas.1922139117

\bibitem[Zhang et al.(2022)]{Zhang2022}
Zhang, S. J., Duzdevich, D., Ding, D., \& Szostak, J. W. (2022).
\textit{Freeze-thaw cycles enable a prebiotically plausible and continuous pathway from nucleotide activation to nonenzymatic RNA copying}.
\textit{Proceedings of the National Academy of Sciences of the United States of America}, 119, e2116429119.
https://doi.org/10.1073/pnas.2116429119

\bibitem[Yeates et al.(2016)]{Yeates2016}
Yeates, J. A. M., Hilbe, C., Zwick, M., Nowak, M. A., \& Lehman, N. (2016).
\textit{Dynamics of prebiotic RNA reproduction illuminated by chemical game theory}.
\textit{Proceedings of the National Academy of Sciences of the United States of America}, 113(18), 5030--5035.
https://doi.org/10.1073/pnas.1525273113

\bibitem[Taylor and Jonker(1978)]{TaylorJonker1978}
Taylor, P. D., \& Jonker, L. B. (1978).
\textit{Evolutionary stable strategies and game dynamics}.
\textit{Mathematical Biosciences}, 40, 145--156.
https://doi.org/10.1016/0025-5564(78)90077-9


\end{thebibliography}
\end{document}